\documentclass[10pt,twocolumn,twoside]{IEEEtran}

\usepackage{amsfonts}
\usepackage{amsbsy}
\usepackage{amssymb}
\usepackage{amscd}
\usepackage[cmex10]{amsmath}
\usepackage{graphicx}
\usepackage[usenames,dvipsnames]{pstricks}
\usepackage{algorithm}
\usepackage[noend]{algpseudocode}
\renewcommand{\algorithmicindent}{.4cm}

\newtheorem{definition}{Definition}

\newtheorem{theorem}{Theorem}
\newtheorem{lemma}{Lemma}
\newtheorem{proposition}{Proposition}
\newtheorem{corollary}{Corollary}

\newcommand{\mat}[1]{\boldsymbol{#1}}

\newcommand{\pp}[1]{{\left( #1 \right)}}

\newcommand{\br}[1]{{\left\{ #1 \right\}}}
\newcommand{\norm}[1]{{ \Vert #1 \Vert }}
\newcommand{\abs}[1]{{ | #1 | }}
\newcommand{\sabs}[1]{{ | #1 |^2 }}
\newcommand{\snorm}[1]{{ \Vert #1 \Vert^2 }}

\def\mrt{{\text{\tiny{MRT}}}}

\def\eig{{\text{\tiny{WF}}}}
\def\zf{{\text{\tiny{ZF}}}}

\def\H{{H}} 
\def\ld{{\log_2}} 
\def\bI{{\mat{I}}} 
\def\bh{{\mat{h}}} 
\def\bw{{\mat{w}}} 
\def\bg1{{g_{11}^{\perp}}}
\def\bg2{{g_{22}^{\perp}}}

\def\coN{{\mathcal{N}}} 
\def\coS{{\mathcal{S}}} 
\def\coT{{\mathcal{T}}} 
\def\coCS{{\mathcal{C}}} 
\def\setA{{\mathcal{A}}} 
\def\setX{{\mathcal{X}}} 
\def\setP{{\mathcal{P}}} 
\def\setR{{\mathcal{R}}} 

\newcommand{\bwzf}[2]{{\mat{w}}^\zf_{{#1}\rightarrow {#2}}}
\newcommand{\bwwf}[2]{{\mat{w}}^\eig_{{#1}\rightarrow {#2}}}

\begin{document}
\title{Coalitional Games in MISO Interference Channels: Epsilon-Core and Coalition Structure Stable Set}%
\author{Rami~Mochaourab,~\IEEEmembership{Member,~IEEE,}
        and~Eduard~Jorswieck,~\IEEEmembership{Senior Member,~IEEE,}
\thanks{\copyright~2014 IEEE. Personal use of this material is permitted. Permission from IEEE must be obtained for all other uses, in any current or future media, including reprinting/republishing this material for advertising or promotional purposes, creating new collective works, for resale or redistribution to servers or lists, or reuse of any copyrighted component of this work in other works.}
\thanks{Part of this work was presented at IEEE CAMSAP, San Juan, Puerto Rico, December 13-16, 2011 \cite{Mochaourab2011b}. Rami Mochaourab is with ACCESS Linnaeus Centre, Signal Processing Department, School of Electrical Engineering, KTH Royal Institute of Technology, 100 44 Stockholm, Sweden. Phone: +4687908434. Fax: +4687907260. E-mail: rami.mochaourab@ieee.org. Eduard Jorswieck is with Communications Theory, Communications Lab, TU Dresden, 01062 Dresden, Germany. E-mail: eduard.jorswieck@tu-dresden.de.}}%
\maketitle
\begin{abstract}
The multiple-input single-output interference channel is considered. Each transmitter is assumed to know the channels between itself and all receivers perfectly and the receivers are assumed to treat interference as additive noise. In this setting, noncooperative transmission does not take into account the interference generated at other receivers which generally leads to inefficient performance of the links. To improve this situation, we study cooperation between the links using coalitional games. The players (links) in a coalition either perform zero forcing transmission or Wiener filter precoding to each other. The $\epsilon$-core is a solution concept for coalitional games which takes into account the overhead required in coalition deviation. We provide necessary and sufficient conditions for the strong and weak $\epsilon$-core of our coalitional game not to be empty with zero forcing transmission. Since, the $\epsilon$-core only considers the possibility of joint cooperation of all links, we study coalitional games in partition form in which several distinct coalitions can form. We propose a polynomial time distributed coalition formation algorithm based on coalition merging and prove that its solution lies in the coalition structure stable set of our coalition formation game. Simulation results reveal the cooperation gains for different coalition formation complexities and deviation overhead models.
\end{abstract}

\IEEEpeerreviewmaketitle
\section{Introduction}
In multiuser interference networks, interference can be the main cause for performance degradation of the systems \cite{Gesbert2010}. With the use of multiple antennas at the transmitters, interference can be managed through cooperative beamforming techniques. For this purpose, backhaul connections are necessary in order to exchange information for cooperation between the transmitters \cite{Gesbert2010}. In this work, we consider cooperation between the transmitters at the beamforming level only, i.e., information concerning the joint choice of beamforming vectors at the transmitters is exchanged between the transmitters but not the signals intended to the users.

The system we consider consists of a set of transmitter-receiver pairs which operate in the same spectral band. The transmitters use multiple antennas while the receivers have single antennas. This setting corresponds to the multiple-input single-output (MISO) interference channel (IFC) \cite{Vishwanath2004}. Next, we describe existing work on beamforming mechanisms in this setting assuming perfect channel state information (CSI) at the transmitters and single-user decoding at the receivers.%
\subsection{Beamforming in the MISO interference channel}%
Optimal beamforming in the MISO IFC corresponds to rate tuples at which it is not possible to strictly improve the performance of the users jointly. Such points are called Pareto optimal. Finding special Pareto optimal operating points in the MISO IFC such as the maximum weighted sum-rate, geometric mean, and proportional-fair rate points are NP-hard problems \cite{Liu2011}. However, finding the max-min Pareto optimal operating point is polynomial time solvable \cite{Liu2011}. In \cite{Qiu2011}, two distributed algorithms are proposed to compute the max-min operating point. After exchanging optimization parameters, the computational load of the beamforming vectors is carried out sequentially at the transmitters. In \cite{Bjornson2011b}, a monotonic optimization framework is proposed to find points such as the maximum sum-rate operating point in general MISO settings with imperfect channel state information at the transmitters. The interested reader is referred to \cite{Bjornson2013} for characterizations of optimal beamforming in MISO settings.

Since optimal beamforming requires high information exchange between the transmitters (or central controller), low complexity and distributed transmission schemes are desirable for practical implementation. When utilizing the reciprocity of the uplink channel in time division duplex (TDD) systems, each transmitter is able obtain perfect local CSI of the channels between itself and all receivers \cite{Bengtsson2001}. Cooperative beamforming schemes based on local CSI do not require CSI exchange through the backhaul connections and are hence favorable \cite{Gesbert2010}.

One cooperative beamforming scheme which requires local CSI is zero forcing (ZF) transmission. This transmission scheme produces no interference at unintended receivers. Heuristic ZF transmission schemes in multicell settings have been proposed in \cite{Brunner2010,Dotzler2011}, where the objective is to efficiently select a subset of receivers at which interference is to be nulled. In \cite{Brunner2010}, the transmitters perform ZF to receivers which are mostly affected by interference. In \cite{Dotzler2011}, a successive greedy user selection approach is applied with the objective of maximizing the system sum-rate. Another beamforming scheme which requires local CSI is Wiener filter (WF) precoding\footnote{Also called minimum mean square error (MMSE) transmit beamforming.} \cite{Joham2005}. Joint WF precoding in MISO IFC is proposed in \cite{Zakhour2009} as a non-iterative cooperation scheme. For the two-user case, the obtained operating point is proven to be Pareto optimal. Furthermore, in \cite{Vazquez2013}, the authors study the reciprocity of the uplink and downlink channels in the MISO IFC to formulate a distributed beamforming scheme in the MISO IFC. In the proposed beamforming schemes in \cite{Zakhour2009,Vazquez2013}, all transmitters cooperate with each other.

Joint cooperation between the links does not necessarily lead to an improvement in the rates of the users compared to a noncooperative and noncomplex beamforming method. Accordingly, a user may not have the incentive to cooperate with all other users. It is then of interest to devise stable cooperative mechanisms that also determine which links would cooperate voluntarily with each other.

Game theory provides appropriate models for designing distributed resource allocation mechanisms. The conflict in multiple antenna interference channels is studied using game theory in \cite{Larsson2009}. The noncooperative operating point (Nash equilibrium) in the MISO IFC corresponds to joint maximum ratio transmission (MRT). This strategy is found to be generally not efficient \cite{Larsson2008a}. In order to improve the performance of the Nash equilibrium, interference pricing is applied in \cite{Schmidt2009}.

Cooperative games in the MISO IFC have been applied in \cite{Nokleby2009, Mochaourab2012, Mochaourab2010}. The Kalai-Smorodinsky solution from axiomatic bargaining theory is studied in the MISO IFC in \cite{Nokleby2009} and an algorithm is provided to reach the solution. In the two-user case, all cooperative solutions, called exchange equilibria, are characterized in \cite{Mochaourab2012} and a distributed mechanism is proposed to reach the Walrasian equilibrium in the setting. Using strategic bargaining, an operating point in the set of exchange equilibria is reached requiring two-bit signaling between the transmitters in \cite{Mochaourab2010}. The approaches in \cite{Mochaourab2012,Mochaourab2010} are however limited to the two-user case.%
\subsection{Applications of Coalitional Games in Partition Form}%
Coalitional games provide structured methods to determine possible cooperation between rational players. A tutorial on the application of coalitional games in communication networks can be found in \cite{Saad2009a}. In these games, a coalition is a set of players which would cooperate to achieve a joint performance improvement. In interference networks, the performance of a coalition of players depends on the coalitions formed outside the coalition. Appropriate in this context are coalitional games in partition form \cite{Thrall1963} which take into account what the players achieve given a \emph{coalition structure}, a partition of the set of users into disjoint cooperative sets.

There are different stability concepts for coalitional games in partition form. In \cite{Apt2006}, \emph{$\mathbb{D}_{hp}$-stability} is proposed which is based on deviation rules of coalition merging and splitting. The stability concept is used for games with transferable utility in \cite{Saad2009} and also applied in \cite{Mochaourab2011b} in the MISO IFC for games with nontransferable utility and partition form. The \emph{recursive core} \cite{Koczy2007} solution concept for coalitional games in partition form has been applied in \cite{Pantisano2013,Guruacharya2013}. In \cite{Pantisano2013}, the set of cooperating base stations (a coalition) performs interference alignment. An algorithm is proposed in which the coalitions can arbitrarily merge and split and proven to converge to an element in the recursive core. In \cite{Guruacharya2013}, MISO channels are considered and the set of cooperating transmitters apply network MIMO techniques. Coalition formation in \cite{Guruacharya2013} is restricted to merging of pairs of coalitions only supporting low complexity implementation.

While in $\mathbb{D}_{hp}$-stability and the recursive core solution concepts a set of players can deviate and form new coalitions, individual based stability \cite{Dreze1980} restricts only a single player to leave a coalition and join another. A deviation in which a user leaves a coalition and joins another if this improves his payoff leads to \emph{Nash stability} \cite{Bogomolnaia2002} and has been used in \cite{Saad2012} in the context of channel sensing and access in cognitive radio. In \cite{Zhou2013}, \emph{individual stability} which is a weaker stability concept than Nash stability is used for coalitional games in the multiple-input multiple-output (MIMO) interference channel. A coalition of links cooperate by performing ZF to each other. Individual stability requires additional to Nash stability the constraint that the payoffs of the members of the coalition in which the deviator wants to join do not decrease.

\subsection{Contributions}
We consider coalitional games without transferable utilities \cite{Osborne1994} among the links. While noncooperative transmission corresponds to MRT, we restrict cooperation between a set of links to either ZF transmission or WF precoding. In \cite{Mochaourab2011b}, the necessary and sufficient conditions for a nonempty core of the coalitional game with ZF transmission are characterized. In this work, we provide the necessary and sufficient conditions for nonempty \emph{strong} and \emph{weak $\epsilon$-core} \cite{Shapley1966} of the coalitional game with ZF beamforming. The $\epsilon$-core generalizes the core solution concept and includes an overhead for the deviation of a coalition. In contrast to the result in \cite{Mochaourab2011b} which specifies an SNR threshold above which the core is not empty, the $\epsilon$-core is not empty above an SNR threshold and also below a specific SNR threshold. These thresholds depend on the deviation overhead measure and the user channels.

While the strong and weak $\epsilon$-core solution concepts consider the stability of the grand coalition only, we study coalitional games in partition form in which several distinct coalitions can form. In \cite{Mochaourab2011b}, coalition formation based on merging and splitting of coalitions has been applied. In this work, we propose a distributed coalition formation algorithm based on coalition merging only. We propose a coalition deviation rule, $q$-Deviation, to incorporate a parameter $q$ which regulates the complexity for finding deviating coalitions. The outcome of the coalition formation algorithm is proven to be inside the \emph{coalition structure stable set} \cite{Diamantoudi2007} of our coalition formation game. Accordingly, the stability of the obtained partition of the links is ensured. We provide an implementation of the coalition formation algorithm and show that only two-bit signaling between the transmitters is needed. Moreover, we prove that the proposed coalition formation algorithm terminates in polynomial time. Simulation results reveal the tradeoff between coalition formation complexity and the obtained performance of the links. In addition, we compare our algorithm regarding complexity and performance to the algorithms in \cite{Guruacharya2013} and \cite{Zhou2013}. 

To the best of our knowledge, the application of the $\epsilon$-core of coalitional games and the coalition structure stable set solution concepts are new for resource allocation in wireless networks. Note, that the coalition structure stable set solution concept for coalitional games in partition form is different than the recursive core \cite{Koczy2007} solution concept used in \cite{Pantisano2013,Guruacharya2013}.

\subsubsection*{Outline}
In Section \ref{sec:preliminaries}, we provide the system and channel model and also describe the noncooperative state of the links. In Section \ref{sec:coalitionalGame}, the game in coalitional form is formulated and its solution is analyzed. In Section \ref{sec:coalitionformation}, we formulate the game in partition form and specify our coalition formation mechanism. We also study the complexity of the proposed deviation model and provide an implementation of the coalition formation algorithm in our considered system. In Section \ref{sec:simResults}, we provide simulation results before we draw the conclusions in Section \ref{sec:conc}.
\subsubsection*{Notations}
Column vectors and matrices are given in lowercase and uppercase boldface letters, respectively. $\norm{\mat{a}}$ is the Euclidean norm of $\mat{a} \in \mathbb{C}^{N}$. $\abs{b}$ and $\abs{\mathcal{S}}$ denote the absolute value of $b \in \mathbb{C},$ and the cardinality of a set $\mathcal{S}$, respectively. $(\cdot)^\H$ denotes the Hermitian transpose. The orthogonal projector onto the null space of $\mat{Z}$ is $\mat{\Pi}_{Z}^{\perp} := \bI - \mat{Z}(\mat{Z}^\H\mat{Z})^{-1}\mat{Z}^\H$, where $\bI$ is an identity matrix. $(x_i)_{i \in \mathcal{S}}$ denotes a profile with the elements corresponding to the set $\mathcal{S}$. The notation $f(x) \in \mathcal{O}(g(x))$ means that the asymptotic growth of $f(x)$ in $x$ is upper bounded by $g(x)$.%

\section{Preliminaries}\label{sec:preliminaries}
\subsection{System and Channel Model}
Consider a $K$-user MISO IFC and define the set of links as $\coN := \{1,...,K\}$. Each transmitter $i$ is equipped with $N_i \geq 2$ antennas, and each receiver with a single antenna. The quasi-static block flat-fading channel vector from transmitter $i$ to receiver $j$ is denoted by $\mat{h}_{ij} \in \mathbb{C}^{N_i \times 1}$. Each transmitter is assumed to have perfect local CSI. The local CSI is gained through uplink training pilot signals \cite{Bengtsson2001}. Here, we assume time division duplex (TDD) systems with sufficiently low delay between the downlink and uplink time slots such that, using channel reciprocity, the downlink channels are estimated to be the same as the uplink channels.

The beamforming vector used by a transmitter $i$ is denoted by $\bw_{i} \in \setA_i$, where the set $\setA_i$ is the \emph{strategy space} of transmitter $i$ defined as
\begin{equation}
\setA_i := \{ \bw \in \mathbb{C}^{N_i \times 1}: \snorm{\bw} \leq 1\},
\end{equation}
where we assumed a total power constraint of one (w.l.o.g.). The basic model for the matched-filtered, symbol-sampled complex baseband data received at receiver $i$ is
\begin{eqnarray}\label{eq:systemmodel}
  y_i = \bh_{ii}^\H \bw_i s_i + \sum\nolimits_{j\neq i} \bh_{ji}^\H \bw_j s_j + n_i,
\end{eqnarray}
\noindent where $s_j \sim \mathcal{CN}(0,1)$ is the symbol transmitted by transmitter $j$ and $n_i\sim \mathcal{CN}(0,\sigma^2)$ is additive white Gaussian noise. We assume that all signal and noise variables are statistically independent. Throughout, we define the SNR as $1/\sigma^2$.

A \emph{strategy profile} is a joint choice of strategies of all transmitters defined as
\begin{equation}\label{eq:strategySpace_all}
(\bw_1,...,\bw_K) \in \setX := \setA_1 \times \cdots \times \setA_K.
\end{equation}

\noindent Given a strategy profile, the achievable rate of link $i$ is
\begin{equation}\label{eq:Rate}
u_i(\bw_1,...,\bw_K) = \ld\pp{1 + \frac{ \sabs{\bh_{ii}^\H \bw_i}}{\sum\nolimits_{j \neq i} \sabs{\bh_{ji}^\H \bw_j} + \sigma^2}},
\end{equation}
\noindent where we assume single-user decoding receivers.

\subsection{Noncooperative Operation}\label{sec:noncoop}
In game theory, games in strategic form describe outcomes of a conflict situation between noncooperative entities. A strategic game is defined by the tuple $\langle \mathcal{N}, \setX, (u_i)_{i \in \coN} \rangle,$ where $\mathcal{N}$ is the set of players (links), $\setX$ is the strategy space of the players given in \eqref{eq:strategySpace_all}, and $u_{k}$ is the utility function of player $k$ given in \eqref{eq:Rate}. In \cite{Larsson2008}, it is shown that maximum ratio transmission (MRT), written for a transmitter $i$ as
\begin{equation}\label{eq:MRT_transmission}
    \bw_i^\mrt = {\bh_{ii}}/{\norm{\bh_{ii}}},
\end{equation}
\noindent is a unique \emph{dominant strategy}. A dominant strategy equilibrium \cite[Definition 181.1]{Osborne1994} of a strategic game is a strategy profile $(\bw^*_1,...,\bw^*_K)$ such that for every player $i \in \mathcal{N}$
\begin{equation}\label{eq:DomStrDef}
    u_i(\bw^*_i,\bw_{-i}) \geq u_i(\bw_i,\bw_{-i}), \quad \forall (\bw_i,\bw_{-i}) \in \setX,
\end{equation}
\noindent where $\bw_{-i}:=(\bw_j)_{j \in \mathcal{N}\setminus\{i\}}$ is the collection of beamforming vectors of all users other than user $i$. Hence, each transmitter chooses MRT irrespective of the strategy choice of the other transmitters. Consequently, the strategy profile $(\bw^\mrt_1,...,\bw^\mrt_K)$ is the unique Nash equilibrium of the strategic game between the links. In \cite{Larsson2008a}, it is shown that joint MRT is near to the Pareto boundary of the achievable rate region in the low SNR regime. In the high SNR regime, joint MRT has poor performance \cite{Larsson2008} while zero forcing (ZF) transmission is near to the Pareto boundary \cite{Larsson2008a}. However, ZF beamforming cannot be implemented if the links are not cooperative. Therefore, we will study cooperative games between the links.%
\section{Coalitional Game}\label{sec:coalitionalGame}
\subsection{Game in Coalitional Form}
In game theory, cooperative games are described by games in coalitional form. A game in coalitional form \cite[Definition 268.2]{Osborne1994} is defined by the tuple
\begin{equation}\label{eq:coalitional_game}
\langle \coN, \setX, V, (u_i)_{i\in \coN} \rangle,
\end{equation}
\noindent where $\coN$ is the set of players, $\setX $ is the set of possible joint actions of the players in \eqref{eq:strategySpace_all}, $V$ assigns to every coalition $\coS$ (a nonempty subset of $\coN$) a set $V(\coS) \subseteq \setX$, and $u_k$ is the utility of player $k$ given in \eqref{eq:Rate}. A coalition $\coS$ is a set of players that are willing to cooperate, and $V(\coS)$ defines their joint feasible strategies. The game considered in \eqref{eq:coalitional_game} is in characteristic form and the mapping $V(\coS)$, called the \emph{characteristic function}, assumes a specific behaviour for the players outside $\coS$.

There exist several models that describe the behavior of the players outside $\coS$ \cite{Marini2007}. For our model, we adopt the $\gamma$-model from \cite{Hart1983} and specify that all players outside a coalition $\coS$ do not cooperate, i.e., build \emph{single-player coalitions}. Later in Section \ref{sec:coalitionformation}, we consider a coalitional game in partition form in which the formation of several coalitions is feasible. Throughout, we assume that the payoff of a player in a coalition cannot be transferred to other players in the same coalition. Thus, we consider games with \emph{nontransferable utilities} which is appropriate for our model in which the achievable rate of one link cannot be utilized at other links.

A solution of the coalitional game is the core which is a set of joint strategies in $\setX$ with which all players want to cooperate in a grand coalition and any deviating coalition cannot guarantee higher utilities to all its members. With this respect, the core strategies are stable. We adopt the following variant of the core \cite{Shapley1966}.\footnote{The definition of $\epsilon$-core in \cite{Shapley1966} is for games with transferable utility. Here we formulate the solution concept for games with nontransferable utility such that the overhead $\epsilon$ is different for each player and not transferable to other players in its coalition.}

\begin{definition}\label{def:sCore}
The \emph{weak $\epsilon$-core} of a coalitional game is the set of all strategy profiles $(\mat{x}_i)_{i\in \coN} \in V(\coN)$ for which there is no coalition $\coS$ and $(\mat{y}_i)_{i\in \coN} \in V(\coS)$ such that $u_i(\mat{y}_1,...,\mat{y}_K) - \epsilon_i > u_i(\mat{x}_1,...,\mat{x}_K)$ with $\epsilon_i \geq 0$ for all $i \in \coS$.
\end{definition}

The weak $\epsilon$-core is not empty if there exists no coalition $\coS \subset \coN$ whose members achieve higher payoff than in the grand coalition $\coN$ taking the additional overhead $\epsilon_i$ in deviation of each player $i\in \coS$ into account. Alternatively, $\epsilon_i$ can be considered as a reward which is given to player $i$ in order to motivate him to stay in the grand coalition. However, since in our model, no external entity is assumed which can give such a reward to the users, the interpretation of $\epsilon_i$ as an overhead is more appropriate. Incorporating the overhead in the solution concept is appropriate in communication networks since coalition deviation requires an additional complexity for searching for possible coalitions to cooperate with. This overhead is discussed later in Section~\ref{sec:complexity} in detail.

In Definition \ref{def:sCore}, the overhead $\epsilon_i$ for deviation of a player $i$ in a coalition $\coS$ is fixed. A stronger notion for the weak $\epsilon$-core accounts for an overhead which depends on the size of the coalition which deviates.
\begin{definition}\label{def:wCore}
The \emph{strong $\epsilon$-core} of a coalitional game is the set of all strategy profiles $(\mat{x}_i)_{i\in \coN} \in V(\coN)$ for which there is no coalition $\coS$ and $(\mat{y}_i)_{i\in \coN} \in V(\coS)$ such that $u_i(\mat{y}_1,...,\mat{y}_K) - \epsilon_i / \abs{\coS} > u_i(\mat{x}_1,...,\mat{x}_K)$ with $\epsilon_i \geq 0$ for all $i \in \coS$.
\end{definition}

The interpretation of the overhead in the definition of the strong $\epsilon$-core originates from the original definition in \cite{Shapley1966} for games with transferable utility where the overhead $\epsilon$ required for the deviating coalition $\coS$ is shared by its members. Hence, the overhead decreases for each member of $\coS$ as the size of $\coS$ increases. This is in contrast to the weak $\epsilon$-core where the overhead is constant for each player. The strong and weak $\epsilon$-core definitions can be regarded as generalizations of the traditional solution concept of the core for which the overhead is set as $\epsilon_i = 0$ for all $i \in \coN$. Interestingly, taking into account the deviation overhead, the solution set of the coalitional game is enlarged, i.e., the core is a subset of the strong $\epsilon$-core \cite{Shapley1966}. Also, the strong $\epsilon$-core is a subset of the weak $\epsilon$-core. For coalitional games in which the core is empty, including the overhead in the deviation could lead to stability of the system.

Next, we will specify the characteristic function $V(\coS)$. While we adopt the $\gamma$-model to assume that the players outside a coalition are noncooperative, in order to define $V(\coS)$ we need to specify the cooperation strategies in a coalition $\coS$. We consider two simple non-iterative transmission schemes which can be applied in a distributed manner. These are ZF and WF beamforming defined in the next subsections.

\subsection{Coalitional Game with Zero Forcing Beamforming}
The transmitters choose MRT if they are not cooperative according to Section \ref{sec:noncoop}. If a transmitter cooperates with a set of links, then it performs ZF in the direction of the corresponding receivers. Hence, we define the mapping
\begin{multline}\label{eq:ZFV}
V^\zf(\coS) = \{(\bw_i)_{i \in \coN} \in \setX: \bw_i = \bwzf{i}{\coS} \text{ for } i \in \coS, \\ \bw_j = \bw_j^\mrt \text{ for } j \in \coN\backslash \coS\},
\end{multline}
\noindent where $\bwzf{i}{\coS}$ is transmitter $i$'s ZF beamforming vector to the links in $\coS$ written as
\begin{equation}\label{eq:beamZFCoal}
\bwzf{i}{\coS} = \frac{\Pi_{\mat{Z}_{i\rightarrow \coS}}^\perp \bh_{ii}}{\norm{\Pi_{\mat{Z}_{i\rightarrow \coS}}^\perp \bh_{ii}}}, \quad \mat{Z}_{i\rightarrow \coS} = (\bh_{ij})_{j\in \coS\backslash \{i\}}.
\end{equation}
\noindent Observe that if the number of antennas $N_i < \abs{\coS}$, then ZF in \eqref{eq:beamZFCoal} is the zero vector, i.e. transmitter $i$ switches its transmission off. Similar to the definition of the strategy profile $V^\zf(\coS)$ in \eqref{eq:ZFV}, it is possible to consider different cooperative transmit beamforming than ZF in a coalition.

According to Definition \ref{def:sCore}, the weak $\epsilon$-core is not empty if and only if
\begin{equation}\label{eq:GrandCoalZF1}
u_i(V^\zf(\coS)) - \epsilon_i \leq u_i(V^\zf(\coN)), \quad \forall i \in \coS \text{~and~} \forall \coS \subset \coN.
\end{equation}
\noindent The next result provides the conditions under which the weak $\epsilon$-core of our game is not empty.
\begin{proposition}\label{thm:grandCoal}
For $\epsilon_i > 0$ for all $i \in \coN$, the weak $\epsilon$-core is not empty if and only if the noise power satisfies $\sigma^2 \leq \bar{\sigma}^2$ and $\sigma^2 \geq \underline{\sigma}^2$ where
\begin{equation}\label{eq:cond_grandCoal1}
\bar{\sigma}^2 := \min_{\coS \subset \coN} \min_{i\in \coS} \left\{ \bar{\sigma}_{i,\coS}^2 \right\}, \quad \underline{\sigma}^2 := \max_{\coS \subset \coN} \max_{i\in \coS} \left\{ \underline{\sigma}_{i,\coS}^2 \right\},
\end{equation}
\noindent with
\begin{equation}\label{eq:sigma_ub}
\bar{\sigma}_{i,\coS}^2 := \left\{
  \begin{array}{ll}
    \infty, & \Delta_{i,\coS} < 0 \text{ or } \Psi_{i,\coS} \geq 0; \\
    \frac{-\Psi_{i,\coS} - \sqrt{\Delta_{i,\coS}} }{2(2^{\epsilon_i} - 1)}, & \Delta_{i,\coS} \geq 0 \text{ and } \Psi_{i,\coS} < 0;
  \end{array}
\right.
\end{equation}
and
\begin{equation}\label{eq:sigma_lb}
\underline{\sigma}_{i,\coS}^2 := \left\{
  \begin{array}{ll}
    0, & \Delta_{i,\coS} < 0 \text{ or } \Psi_{i,\coS} \geq 0; \\
    \frac{-\Psi_{i,\coS} + \sqrt{\Delta_{i,\coS}} }{2(2^{\epsilon_i} - 1)}, & \Delta_{i,\coS} \geq 0 \text{ and } \Psi_{i,\coS} < 0;
  \end{array}
\right.
\end{equation}
and the used parameters are defined as
\begin{subequations}
\begin{align}\label{eq:Psi}
\Delta_{i,\coS} & := \Psi_{i,\coS}^2 - 4 (2^{\epsilon_i} - 1)2^{\epsilon_i} C_{i} B_{i,\coS},\\ \label{eq:Delta}
\Psi_{i,\coS} & := \pp{2^{\epsilon_i} (B_{i,\coS} + C_{i}) - (B_{i,\coS} + A_{i,\coS})},\\ \label{eq:A}
A_{i,\coS} & := {\sabs{\bh_{ii}^\H \bw_{i\rightarrow\coS}^\zf}},\\ \label{eq:B}
B_{i,\coS} & := \sum\nolimits_{j \in \coN \backslash \coS} \sabs{\bh_{ji}^\H \bw_j^\mrt},~ C_{i} := \sabs{\bh_{ii}^\H \bw_{i\rightarrow\coN}^\zf}.
\end{align}
\end{subequations}
\end{proposition}
\begin{IEEEproof}
The proof is provided in Appendix \ref{proof:grandCoal}.
\end{IEEEproof}

Proposition \ref{thm:grandCoal} implies that for all SNR values $1/\sigma^2 \geq 1/\bar{\sigma}^2$ and $1/\sigma^2 \leq 1/\underline{\sigma}^2$ it is profitable for all players to jointly perform ZF. In the case where the overhead $\epsilon_i = 0$ for all players, the conditions under which the core is not empty have been given in \cite[Proposition 1]{Mochaourab2011b} and restated below as a special case of Proposition \ref{thm:grandCoal}.

\begin{corollary}\label{thm:zeroOverhead_core}
For $\epsilon_i = 0$ for all $i \in \coN$, the weak $\epsilon$-core is not empty if and only if $\sigma^2\leq \hat{\sigma}^2$ where
\begin{equation}\label{eq:cond_grandCoal_constants}
\hat{\sigma}^2 := \min_{\coS \subset \coN} \min_{i\in \coS} \br{ {B_{i,\coS} C_{i}}/({A_{i,\coS} - C_{i}})}.
\end{equation}
\end{corollary}

Interestingly, in comparison to the core without deviation overhead, the weak $\epsilon$-core is not empty above an SNR threshold and also below an SNR threshold. The weak $\epsilon$-core is also not empty at low SNR is due to the fact that the noise power at low SNR is much larger than the interference. Then, the performance difference between joint ZF beamforming in the grand coalition compared to the performance of another beamforming strategy is not large enough to compensate for the overhead leading to the formation of the grand coalition.

The derivation of the conditions for nonempty strong $\epsilon$-core in Definition \ref{def:wCore} is analogous to that in Proposition \ref{thm:grandCoal} in which for a player $i \in \coS$, the term $\epsilon_i$ is replaced with $\epsilon_i / \abs{\coS}$. It must be noted that in order to calculate the conditions for nonempty weak and strong $\epsilon$-core in Proposition \ref{thm:grandCoal}, an exhaustive search over $2^\abs{\coN} - 1$ nonempty subsets of $\coN$ must be performed.

\begin{figure}[t]
  \centering
  \includegraphics[width=8.8cm,clip]{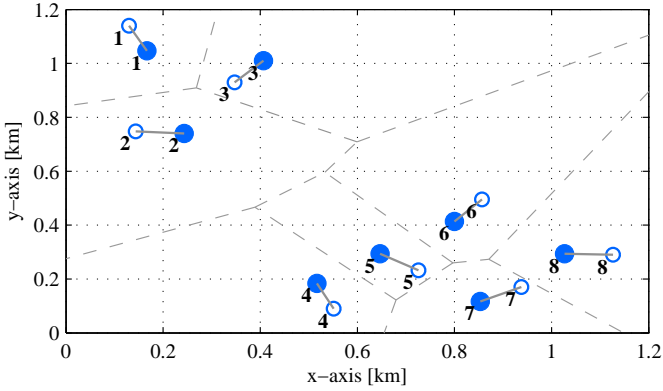}
  \caption{A setting with $8$ links and each transmitter uses $8$ antennas. In order to include the effect of distances between the links on the received power gains we use the following path loss model: Let $d_{k\ell}$ be the distance between transmitter $k$ and a receiver $\ell$ in meters and $\delta$ be the path loss exponent, we write the channel vector $\bh_{k\ell} = d^{-\delta/2}_{k\ell}\tilde{\bh}_{k\ell}/\norm{\tilde{\bh}_{k\ell}}$ with $\tilde{\bh}_{k\ell} \sim \mathcal{CN}(0,\mat{I})$. We define the SNR as SNR$ = d_{kk}^{-\delta} / \sigma^2$ and we set $\delta = 3$.}\label{fig:topo}
\end{figure}

In \figurename~\ref{fig:emptycore}, we plot the conditions for empty strong and weak $\epsilon$-core from Proposition~\ref{thm:grandCoal} for the setting in \figurename~\ref{fig:topo}. Only for an overhead strictly larger than zero does a lower SNR threshold exists for nonempty strong and weak $\epsilon$-core. As the overhead increases, the lower threshold $1/\bar{\sigma}^2$ increases and the upper threshold $1/\underline{\sigma}^2$ decreases. Consequently, the SNR region where all players have an incentive to jointly perform ZF transmission becomes larger. It is shown that if the conditions for the weak $\epsilon$-core not to be empty are satisfied then they are also satisfied for the strong $\epsilon$-core. The operation of wireless systems is usually in the range between $5$ and $20$ dB SNR. It can be seen from \figurename~\ref{fig:emptycore} that the conditions for the stability of the grand coalition with ZF beamforming requires relatively higher overhead measure at the links.

There is a relation between the result in Proposition \ref{thm:grandCoal} and the notion of \emph{cost of stability} \cite{Bachrach2009}. In \cite{Bachrach2009}, it is assumed that all users have the same $\epsilon_i$. The cost of stability specifies the smallest overhead $\epsilon_i$ such that the weak $\epsilon$-core is not empty. \figurename~\ref{fig:emptycore} illustrates the cost of stability which corresponds to the overhead $\epsilon$ on the boundary points of the region where the weak $\epsilon$-core is empty. That is, for a fixed SNR value, the boundary point is the smallest overhead value with which the weak $\epsilon$-core is not empty.

\begin{figure}[t]
  \centering
  \includegraphics[width=8.8cm,clip]{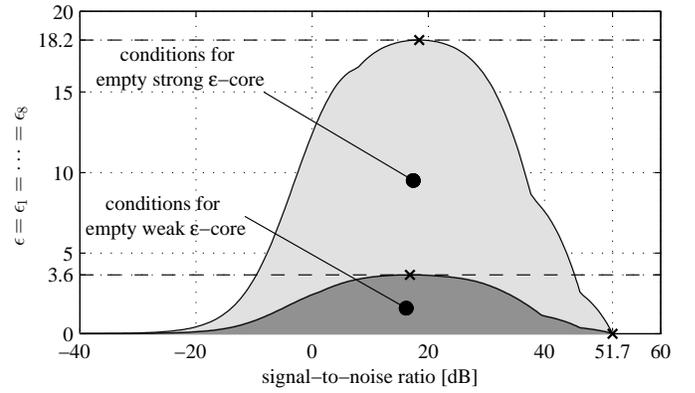}
  \caption{Conditions for empty strong and weak $\epsilon$-cores are plotted in the filled regions for the setting with $8$ links in \figurename~\ref{fig:topo}.}\label{fig:emptycore}
\end{figure}

In \figurename~\ref{fig:emptycore}, it is shown that above a certain overhead level, the $\epsilon$-core is nonempty for any SNR value. This overhead level is obtained during the proof of Proposition \ref{thm:grandCoal} and stated here.
\begin{corollary}\label{thm:globalgrandcoal}
The weak $\epsilon$-core is not empty for any $\sigma^2 > 0$ if and only if $\Delta_{i,\coS} < 0$ or $\Psi_{i,\coS} \geq 0$, for all $i\in \coS, \coS \subset \coN$, where $\Delta_{i,\coS}$ and $\Psi_{i,\coS}$ are given, respectively, in \eqref{eq:Psi} and \eqref{eq:Delta} in Proposition \ref{thm:grandCoal}.
\end{corollary}

From Corollary \ref{thm:zeroOverhead_core}, in the case where the overhead $\epsilon_i$ is zero for all $i \in \coN$, the core is not empty only above an SNR threshold. Next, we show that also in this case, a player does not have an incentive to build a coalition with another player at low SNR, i.e.,
\begin{equation}\label{eq:singleCoal}
u_i(V^\zf(\{i\})) > u_i(V^\zf(\coS)), \forall i \in \coS, \forall \coS \subseteq \coN, \abs{\coS} > 1.
\end{equation}
Notice that $V^\zf(\{i\}) = (\bw^\mrt_1,...,\bw^\mrt_K)$. The conditions for \eqref{eq:singleCoal} to hold are given in \cite[Proposition 2]{Mochaourab2011b} and restated here.
\begin{proposition}\label{thm:singleCoal} Single-player coalitions exist if
\begin{equation}\label{eq:cond_singleCoal}
\sigma^2 > \check{\sigma}^2:=\max_{\coS \subseteq \coN} \max_{i\in \coS} \br{ \frac {A_{i,\coS} B_{i,\{i\}} -\snorm{\bh_{ii}} B_{i,\coS}}{\snorm{\bh_{ii}} - A_{i,\coS} }},
\end{equation}
\noindent with $A_{i,\coS}$ and $B_{i,\coS}$ defined in \eqref{eq:A} and \eqref{eq:B}, respectively. 
\end{proposition}

\subsection{Coalitional Game with Wiener Filter Precoding}
In this section, we assume the players cooperate by performing WF precoding with the players in their own coalition. For link $i$ in coalition $\coS$, transmitter $i$'s WF precoding is
\begin{equation}\label{eq:beamEIGCoal}
\bwwf{i}{\coS} = \frac{({\bI\sigma^2 + \sum_{j\in \coS\backslash \{i\}} \bh_{ij} \bh_{ij}^\H})^{-1}\bh_{ii}}{\norm{({\bI\sigma^2 + \sum_{j\in \coS\backslash \{i\}} \bh_{ij} \bh_{ij}^\H})^{-1}\bh_{ii}}}.
\end{equation}
In comparison to ZF beamforming in \eqref{eq:beamZFCoal}, WF precoding is suitable when the number of antennas at the transmitters is smaller than the number of links in the coalition. WF precoding in \eqref{eq:beamEIGCoal} has interesting behavior for asymptotic SNR cases \cite{Joham2005}. In the high SNR regime ($\sigma^2 \rightarrow 0$), $\bwwf{i}{\coS}$ converges to $\bwzf{i}{\coS}$ in \eqref{eq:beamZFCoal}. In the low SNR regime ($\sigma^2 \rightarrow \infty$), $\bwwf{i}{\coS}$ converges to $\bw^\mrt_i$ in \eqref{eq:MRT_transmission}.

The game in coalitional form with WF precoding is $\langle \coN, \setX, V^\eig, (R_k)_{k\in \coN} \rangle$ where the mapping $V^\eig$ which defines the strategy profile according to WF cooperation scheme is
\begin{multline}
V^\eig(\coS) = \{(\bw_k)_{k \in \coN} \in \setX: \bw_k = \bwwf{k}{\coS} \text{ for } k \in \coS, \\ \bw_\ell = \bw_\ell^\mrt \text{ for } \ell \in \coN\backslash \coS\}.
\end{multline}
Conditions for nonempty strong and weak $\epsilon$-core of the coalitional game with WF precoding in terms of SNR thresholds are hard to characterize because the noise power in \eqref{eq:beamEIGCoal} is inside the matrix inverse. 

Next, we will study coalition formation games between the links and use the WF and ZF beamforming schemes as cooperative beamforming strategies.%
\section{Coalition Formation}\label{sec:coalitionformation}
In the previous section, we have used the strong and weak $\epsilon$-core as solutions to our coalitional game which only consider the feasibility of the formation of the grand coalition. In this section, we enable the formation of several disjoint coalitions to form a coalition structure. A \emph{coalition structure} $\coCS$ is a partition of $\coN$, the grand coalition, into a set of disjoint coalitions $\{\coS_1,...,\coS_L\}$ where $\bigcup^L_{j = 1} \coS_j = \coN$ and $\bigcap^L_{j = 1}\coS_j = \emptyset$.

Let $\setP$ denote the set of all partitions of $\coN$. We consider the following game in partition form \cite{Thrall1963}:\footnote{In \cite{Thrall1963}, the game in partition form is represented by $\langle \coN, U \rangle$ where $U:\setP \rightarrow \mathbb{R}^K$. We change the notation to be analogous to the coalitional game formulation in \eqref{eq:coalitional_game}.}
\begin{equation}
\langle \coN, \setX, F, (u_i)_{i\in \coN} \rangle,
\end{equation}
\noindent where $F: \setP \rightarrow \setX$ is called the \emph{partition function}. We consider two scenarios for player cooperation in a coalition. The scenarios correspond to ZF and WF transmissions. Given a coalition structure $\coCS$, the strategy profile of the players according to ZF or WF is defined by
\begin{equation}
F^{\text{bf}}(\coCS) := \{(\bw_i)_{i \in \coN} \in \setX: \bw_i = \bw_{i\rightarrow \coS_j}^{\text{bf}} \text{ for } i \in \coS_j, \coS_j \in \coCS\},
\end{equation}
\noindent where $\text{bf} = \{\text{\small{ZF},\small{WF}}\}$ with $\bwzf{i}{\coS_j}$ and $\bwwf{i}{\coS_j}$ defined in \eqref{eq:beamZFCoal} and \eqref{eq:beamEIGCoal}, respectively. Notice that if $\abs{\coS_j}=1$ and $i \in \coS_j$, then $\bwzf{i}{\coS_j} = \bwwf{i}{\coS_j} = \bw_i^\mrt$. For a coalition structure $\coCS$, $F^\zf(\coCS)$ is a strategy profile in which each player chooses ZF to the players in his coalition. Similarly, $F^\eig(\coCS)$ is the strategy profile when WF is applied. In our case, the coalition structure uniquely determines the associated strategy of each player. Consequently, the payoff of each player is directly related to the formed coalition structure.

The payoffs of the members of a coalition depend on which coalitions form outside. The effects caused by other coalitions on a specific coalition are called \emph{externalities} \cite{Yi1997}. In our game, due to the interference coupling between the links, externalities exist. These are categorized under negative and positive externalities. In the case of negative externalities, the merging of two coalitions reduces the utility of the other coalitions. While positive externalities lead to an increase in the payoff of other coalitions when two coalitions merge. In our case, if two coalitions merge, both types of externalities can occur.

The dynamics that lead to a specific coalition structure are the study of coalition formation games \cite{Aumann1974,Marini2007}. We are interested in coalition structures which are stable. The main steps to describe the dynamics of coalition formation to reach a stable coalition structure are the following: First, we must specify a \emph{deviation rule}. This rule describes the feasible transition from one coalition structure to the next. The second step is to define a \emph{comparison relation} between different coalition structures. Accordingly, a feasible deviation from one coalition structure to the next is acceptable if this leads to a preferable coalition structure. Afterwards, the stability of the coalition formation process must be studied. For this purpose, a stability concept for coalition structures must be specified.

A set $\coT$ of at most $q$ coalitions in some arbitrary coalition structure $\coCS_0 \in \setP$ can \emph{merge} to form a single coalition. In doing so, the coalition structure $\coCS_0$ changes to $\coCS_1 \in \setP$. We formally define this mechanism as follows.
\begin{definition}[$q$-Deviation]\label{def:coalition_deviation}
The notation $\coCS_0 \overset{q,\coT}{\longrightarrow} \coCS_1$ indicates that the coalitions in $\coT$, where $\coT \subset \coCS_0 \in \setP$ and $\abs{\coT} \leq q$ merge to form the coalition $\coS = \bigcup \coT$. Then, the coalition structure $\coCS_0$ changes to $\coCS_1 = \coCS_0 \setminus \coT \cup \coS$ in the set of coalition structures $\setP$.
\end{definition}
The motivation behind the merging deviation model in Definition \ref{def:coalition_deviation} is that coalition formations starts from the noncooperative state of single-player coalitions and hence cooperation requires merging of coalitions. The deviation complexity in Definition \ref{def:coalition_deviation} is tunable through the parameter $q$. The larger $q$ is, the more complex it is to search for possible merging coalitions since the number of possible combinations increases with $q$. We study the complexity for this search in Section~\ref{sec:complexity}.

%
\begin{algorithm}[t]
\caption{\label{alg:coalition0} Coalition formation algorithm.}
\begin{algorithmic}[1]
\State \textbf{Input}: {$\coN$, $(\epsilon_1,\ldots,\epsilon_K)$, $q$}
\State \textbf{Initilize}: $k = 0$, $\coCS_{0} = \br{\br{1},\ldots,\br{K}}$
\For{$\coT \subset \coCS_k, \abs{\coT} \leq q$}
\State $\coCS_k \overset{q,\coT}{\longrightarrow} \coCS_{k+1}$;
\If{$\coCS_k \prec_\coT \coCS_{k+1}$}
\State $k = k+1$;
\State Go to Step 2;
\EndIf
\EndFor
\State \textbf{Output}: {$\coCS_k$}
\end{algorithmic}
\end{algorithm}

Given a coalition structure, we assume that the coalitions can communicate with each other in order to find possible performance improvement through deviation. A deviating set of coalitions $\coT$ according to $q$-Deviation in Definition \ref{def:coalition_deviation} can compare the resulting coalition structure $\coCS_1$ with the previous coalition structure $\coCS_0$ by the \emph{Pareto dominance relation} $\prec_\coT$ specified as follows:
\begin{equation}
\begin{split}\label{eq:preference}
\coCS_0 \prec_\coT \coCS_1 \Leftrightarrow &\\
\forall i \in \bigcup \coT &: u_i({F}^\text{bf}(\coCS_0)) - \epsilon_i \leq u_i({F}^\text{bf}(\coCS_1)), \text{ and }\\
\exists j \in \bigcup \coT &: u_j({F}^\text{bf}(\coCS_0)) - \epsilon_i < u_j({F}^\text{bf}(\coCS_1)),
\end{split}
\end{equation}
\noindent with $\text{bf} = \{\text{\small{ZF},\small{WF}}\}$. As in the definition of the weak and strong $\epsilon$-core in the previous section, we include the overhead $\epsilon_i \geq 0$ in \eqref{eq:preference} for a deviating player $i$. In \eqref{eq:preference}, the overhead $\epsilon_i$ is subtracted from the current utility in $\coCS_0$ when comparing it to the utility in the new coalition structure $\coCS_1$. The motivation for this model is as follows: the overhead required for the utility comparison and communication with the members of $\coT$ should make coalition merging more attractive. Later in \eqref{eq:alpha} and \eqref{eq:alpha2} in Section~\ref{sec:simResults}, we specify two different overhead models for performance comparison.

According to \eqref{eq:preference}, the notation $\coCS_0 \prec_\coT \coCS_1$ means that each player in $\coT$ prefers the coalition structure $\coCS_1$ to $\coCS_0$. Note, that $\coCS_0 \prec_\coT \coCS_1$ indicates that the coalition structure $\coCS_1$ is preferred to $\coCS_0$ by the players in $\coT$ and the preferences of the players in $\coN\setminus\coT$ are not considered in the comparison. This choice enables the set of deviating coalitions $\coT$ to decide on their own if they want to merge without consulting the remaining players. However, note that the strategies of deviating coalitions do affect the utilities of all players.

Based on the deviation rule in Definition \ref{def:coalition_deviation} and the coalition preference in \eqref{eq:preference}, we formulate a binary relation to compare two coalition structures.
\begin{definition}[$q$-Dominance]\label{def:dominance}
The coalition structure $\coCS_1$ $q$-dominates $\coCS_0$, written as $\coCS_1 \gg^q \coCS_0$, if there exists a set of coalitions $\coT \subset \coCS_0$ such that $\coCS_0 \overset{q,\coT}{\longrightarrow} \coCS_{1}$, and $\coCS_1 \prec_{\coT} \coCS_{0}$.
\end{definition}

According to the previous definitions, we define the coalition formation game as $(\setP,\gg^q)$, where $\setP$ is the set of all coalition structures and $\gg^q$ is the dominance relation defined in Definition \ref{def:dominance}. The solution of the coalition formation game $(\setP,\gg^q)$ is a set of coalition structures with special stability properties. We use the \emph{coalition structure stable set} as a solution concept for $(\setP,\gg^q)$ adopted from \cite[p. 110]{Diamantoudi2007}.
\begin{definition}[Coalition Structure Stable Set]\label{def:stable_set}
The set of coalition structures $\setR \subset \setP$ is a coalition structure stable set of $(\setP,\gg^q)$ if $\setR$ is both internally and externally stable:
\begin{itemize}
\item $\setR$ is \emph{internally stable} if there do not exist $\coCS, \coCS' \in \setR$ such that $\coCS \gg^q \coCS'$,
\item $\setR$ is \emph{externally stable} if for all $\coCS \in \setP \setminus \setR$ there exists $\coCS' \in \setR$ such that $\coCS' \gg^q \coCS$.
\end{itemize}
\end{definition}
The coalition structure stable set for coalition formation games is a modification of the stable set solution concept of coalitional games in characteristic form which has been originally proposed in \cite{Neumann1944}. Unlike the core, the stable set is not necessarily unique.

The coalition structure stable set of $(\setP,\gg^q)$ has at least one element which is the grand coalition $\coN$. As $q$-Deviation (Definition \ref{def:coalition_deviation}) only allows merging of coalitions, the coalition structure $\coN$ is stable because no deviation is possible. We assume however that the players start their operation in the noncooperative state corresponding to the Nash equilibrium (Section \ref{sec:noncoop}). In order to reach a solution in the coalition structure stable set, we provide Algorithm \ref{alg:coalition0}.\footnote{Other initializing coalition structures can also be supported by Algorithm \ref{alg:coalition0} which may lead to different outcomes.} Step 4 finds a $q$-Deviation (Definition \ref{def:coalition_deviation}) by searching over all coalition merging possibilities given coalition structure $\coCS_k$ and $q$, the maximum number of coalitions that are allowed to merge. For a possible deviating coalition $\coT$, the resulting coalition structure $\coCS_{k+1}$ is compared with the previous coalition according to $q$-Dominance (Definition \ref{def:dominance}). If Step 5 is true, then the new coalition structure $\coCS_{k+1}$ is adopted and the index $k$ is incremented. If no deviating coalitions satisfy Step 5, then the algorithm terminates.

\begin{proposition}
Algorithm \ref{alg:coalition0} converges to a coalition structure in the coalition structure stable set of $(\setP,\gg^q)$.
\end{proposition}
\begin{proof}
First, the algorithm is guaranteed to converge since only merging operations are allowed and the number of merging steps is finite. We prove by contradiction that the resulting coalition structure is in the coalition structure stable set. Let the solution of the algorithm be $\coCS_k$. If $\coCS_k$ is outside the coalition structure stable set of $(\setP,\gg^q)$, then according to external stability, there exists a coalition structure $\coCS'$ in the coalition structure stable set such that $\coCS' \gg^q \coCS_k$. However, if this is the case, $\coCS_k$ would not be a solution of Algorithm~\ref{alg:coalition0} since Algorithm~\ref{alg:coalition0} terminates when no coalition structure is found that $q$-dominates the obtained coalition structure. If $\coCS_k$ is inside the coalition structure stable set of $(\setP,\gg^q)$ and there is a coalition structure $\coCS''$ in the stable set such that $\coCS'' \gg^q \coCS_k$, then also $\coCS_k$ would not be a solution of Algorithm \ref{alg:coalition0}. Accordingly, the coalition structure $\coCS_k$ resulting from Algorithm \ref{alg:coalition0} must be in the coalition structure stable set of $(\setP,\gg^q)$.
\end{proof}

Next, we discuss the complexity for finding a set of coalitions which merge according to $q$-Deviation in Definition~\ref{def:coalition_deviation} and compare it to an analogous rule based on coalition splitting.

\subsection{Complexity}\label{sec:complexity}
\begin{figure}[t]
  \centering
  \includegraphics[width=8.8cm,clip]{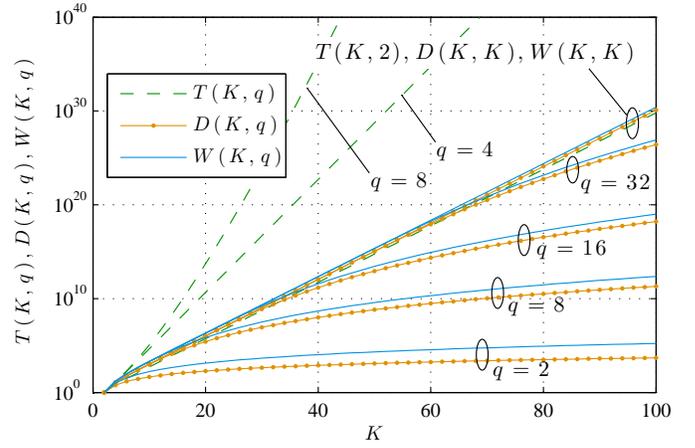}
  \caption{Illustration of the measures ${D}(K,q)$, $T(K,q)$, and $W(K,q)$ defined in \eqref{eq:comp_merging}, \eqref{eq:comp_spliting}, and \eqref{eq:complexity} respectively.}\label{fig:complexity}
\end{figure}
Our coalition formation algorithm is influenced by the following works on coalitional games in partition form \cite{Ray1997, Diamantoudi2007}. In \cite{Ray1997} a solution concept called \emph{equilibrium binding agreements} is proposed. The algorithm starts in the grand coalition, and only splitting operations can occur. The stability of a coalition structure is based on finding a sequence of splitting deviations which achieve higher utilities in the final coalition structure. Moreover, the players are considered farsighted, i.e., can anticipate the effects of their actions to the actions of other players. This mechanism is however shown to be inefficient and has motivated the extension in \cite{Diamantoudi2007}. In \cite{Diamantoudi2007}, a coalition formation algorithm is proposed where the splitting operation proposed in \cite{Ray1997} is adapted such that the coalitions that split can also merge in an arbitrary manner. Despite the increased complexity, Pareto efficiency is achieved only in special cases.

The reason for choosing coalition deviation based only on merging of coalitions in our work is twofold: First, since the users start their operation in the noncooperative state of single-player coalitions, coalition formation must be able to merge coalitions. Second, the splitting operation is far more complex than the merging operation, as discussed next.

According to Definition \ref{def:coalition_deviation}, the number of possible ways to merge a set of at least two and at most $q$ coalitions from the coalition structure $\coCS$ is
\begin{equation}\label{eq:comp_merging}
{D}(\abs{\coCS},q) = \sum\nolimits_{j=2}^{\min\{q,\abs{\coCS}\}} \binom{\abs{\coCS}}{j}.
\end{equation}
The expression above has no closed form. For the special case $q = \abs{\coCS}$, ${D}(\abs{\coCS},\abs{\coCS}) = 2^{\abs{\coCS}} - \abs{\coCS} - 1$ and for $q = 2$ we get ${D}(\abs{\coCS},2) =({\abs{\coCS}}^2 - {\abs{\coCS}})/2$. The worst case complexity corresponds to single-player coalitions because then the number of coalitions in $\coCS$ is largest and equal to $K$. This is the initial coalition structure which starts Algorithm \ref{alg:coalition0}.
\begin{lemma}\label{lem:complexity}
The growth rate of ${D}(\abs{\coCS},q)$ in \eqref{eq:comp_merging} is bounded by $\mathcal{O}(K^q)$ for fixed $q$.
\end{lemma}
\begin{IEEEproof}
First, observe that ${D}(\abs{\coCS},q) \leq {D}(K,q)$ for $\abs{\coCS} \leq K$. According to the definitions in \cite[24.1.1]{Abramowitz1972}, and for $q \leq K$ we can calculate
\begin{align}
{D}(K,q) &= \sum\nolimits_{j=2}^{q} \binom{K}{j} = \sum\nolimits_{j=2}^{q} \binom{K}{K-j} \\ \nonumber
& = \frac{K(K-1)}{2!} + \frac{K(K-1)(K-2)}{3!} + \cdots \\
& \quad + \frac{\overbrace{K(K-1)\cdots (K - q + 1)}^\text{$q$ factors}}{q!}.
\end{align}
Accordingly, ${D}(K,q) \in \mathcal{O}(K^q)$.
\end{IEEEproof}

In comparison to the merging rule, the number of combinations for splitting a set $\coS$ to $k, k \leq \abs{\coS},$ subsets is given by the \emph{Stirling number of the second kind} \cite[Theorem 8.2.6]{Brualdi2004}:
\begin{equation}
{S}(\abs{\coS},k) = \frac{1}{k!} \sum\nolimits_{t=0}^{k}(-1)^{t}\binom{k}{t} (k - t)^{\abs{\coS}}.
\end{equation}
If we allow, as we did in the merging case, the splitting of $\coS$ into at least two and at most $q$ subsets, we get the following number of combinations:
\begin{equation}\label{eq:comp_spliting}
T(\abs{\coS},q) = \sum\nolimits_{k=2}^{q} S(\abs{\coS},k).
\end{equation}
For $q = \abs{\coS}$, the number above corresponds to the $\abs{\coS}^\text{th}$ Bell number minus one.\footnote{The $\abs{\coS}^\text{th}$ Bell number is $\sum_{k=0}^{q} S(\abs{\coS},k)$ and ${S}(\abs{\coS},0) + {S}(\abs{\coS},1)) = 1$.} In \figurename~\ref{fig:complexity}, we compare the complexity of the merging and splitting operations for increasing number of users $K$ and different values of $q$. It can be observed that the splitting operation requires much more searching combinations than the merging operation.

\subsection{Implementation in the MISO interference channel}
In Algorithm \ref{alg:coalition}, we provide an implementation of Algorithm \ref{alg:coalition0} in the MISO setting. The algorithm is initialized according to the Nash equilibrium, i.e., all coalitions are singletons. The term $r = \min\{q,\abs{\coCS_0}\}$ is the number of coalitions that are allowed to merge. The quantity $\Theta$ is initialized to zero and will aggregate the total number of utility comparisons during the algorithm. This measure will be used later in the simulations in Section~\ref{sec:simResults} to reveal numerically the complexity of the algorithm.

\begin{algorithm}[t]
\caption{\label{alg:coalition} Implementation of Algorithm \ref{alg:coalition0}.}
\begin{algorithmic}[1]
\State \textbf{Input}: {$\coN$, $(\epsilon_1,\ldots,\epsilon_K)$, $q$, $\text{bf} = \{\text{\small{ZF}, \small{WF}}\}$}
\State \textbf{Initialize}: $k = 0$, $\coCS_{0} = \br{\br{1},\ldots,\br{K}}$, $n^\text{iter}$, $r = \min\{q, \abs{\coCS_{0}}\}$, $\Theta = 0$
\While{$r\geq2$ \textbf{and} $\abs{\coCS_k} > 2$}
\State \parbox[t]{\dimexpr\linewidth-\algorithmicindent}{Each user generates lexicographically ordered r-combinations of $\coCS_k$: $\{\coT_{1},\ldots,\coT_{\binom{\abs{\coCS_k}}{r}}\}$; \strut}
\For{$\ell = 1:\binom{\abs{\coCS_k}}{r}$}
\State $n^\text{iter} = n^\text{iter} + 1;$
\State \parbox[t]{\dimexpr\linewidth-\algorithmicindent-\algorithmicindent}{Each user in $\coT_\ell$ temporarily generates $\coCS_{k+1}$ from $\coCS_{k}$ by merging $\coT_\ell$; \strut}
\State \parbox[t]{\dimexpr\linewidth-\algorithmicindent-\algorithmicindent}{Each user $i \in \coT_\ell$ compares his utility $u_i({F}^\text{bf}(\coCS_k)) - \epsilon_i$ to $u_i({F}^\text{bf}(\coCS_{k+1}))$; \strut}
\State \parbox[t]{\dimexpr\linewidth-\algorithmicindent-\algorithmicindent}{Increment the number of utility comparisons: $\Theta = \Theta + \sum_{\coS \in \coT_\ell} \abs{\coS}$; \strut}
\State \parbox[t]{\dimexpr\linewidth-\algorithmicindent-\algorithmicindent}{Each user in $\coT_\ell$ sends a message to the other users in $\coT_\ell$ from the set \{M1, M2, M3\}; \strut}
\If{all messages are M1 or M2 and some M1}
\State {Users in $\coT_\ell$ merge to form a single coalition; \strut}
\State {Users in $\coT_\ell$ send M4 to users outside $\coT_\ell$; \strut}
\State {$k = k+1$ and $r = \min\{q, \abs{\coCS_{k}}\}$; \strut}
\State {Go to Step 4; \strut}
\EndIf
\EndFor
\State $r = r - 1;$
\EndWhile
\State \textbf{Output}: {$\coCS_k, \Theta$}
\end{algorithmic}
\end{algorithm}%

We assume that when a coalition structure forms, such as $\coCS_k = \{\coS_1,\ldots,\coS_L\}$, each coalition is given a unique index which is commonly known to all users. Moreover we assume that each user knows the members of each coalition $\coS_i, i = 1,\ldots,L$. In Step 4 in Algorithm \ref{alg:coalition}, each user generates a list of $r$-subsets of the indexes $\{1,\ldots,L\}$ of the coalition structure $\coCS_k$ and these subsets are ordered in \emph{lexicographic order}. The generation of this list can be done using the algorithm provided in \cite[Section 4.4]{Brualdi2004}. In this manner, each user has the same ordered list of the subsets of $\{1,\ldots,L\}$ of size $r$. We assume that all users are synchronized in the sense that the users consider the same element $\coT_\ell$ of the generated list for a period of time in which negotiation takes place. Then, in Step 7 in Algorithm \ref{alg:coalition}, the new coalition structure $\coCS_{k+1}$ is temporarily formed by merging $\coT_\ell$. In Step 8, each user in $\coT_\ell$ evaluates his utility in the new coalition. In Step 9, the number of utility comparisons is incremented according to the number of users in $\coT_\ell$. Following Step 8 in which the utility comparisons are made, in Step 10 each user in $\coT_\ell$ communicates one of the following messages to the other users in $\coT_\ell$:\medskip

\begin{itemize}
\item (M1) $\leftarrow$ utility improves
\item (M2) $\leftarrow$ utility is the same
\item (M3) $\leftarrow$ utility decreases
\item (M4) $\leftarrow$ coalition forms
\end{itemize}
\medskip

Note that any of the above four messages can be exchanged between two links requiring two bits of information. In Step 11, the Pareto dominance relation defined in \eqref{eq:preference} is examined. If the condition in Step 11 is true then the coalitions in $\coT_\ell$ merge and all users outside $\coT_\ell$ are informed of the new coalition structure by exchanging the message M4. Algorithm \ref{alg:coalition} terminates if the grand coalition is formed or if $r$, the number of coalitions to merge, is less than two. 

Algorithm \ref{alg:coalition} provides a structured method to find a possible mergings between coalitions by exploiting the algorithm in \cite[Section 4.4]{Brualdi2004}. The complexity of Algorithm \ref{alg:coalition} is however related to the complexity of an exhaustive search with the restriction on the maximum number of coalitions that can merge, $q$.

\begin{theorem}\label{thm:complexity}
The number of iterations of Algorithm \ref{alg:coalition} is bounded as $n^\text{iter} \in \mathcal{O}(K^q)$.
\end{theorem}
\begin{IEEEproof}
In the worst case, only two coalitions merge at a time, i.e. with $r=2$, and the two coalitions which merge are according to the \emph{last element} in the r-combination list generated in Step 4 in Algorithm \ref{alg:coalition}. In addition, in worst case this described behaviour occurs for every merging occasion until the grand coalition forms. Accordingly, the worst case number of iterations of Algorithm \ref{alg:coalition} is given as:
\begin{equation}\label{eq:complexity}
n^\text{iter} \leq W(K,q) := \sum\nolimits_{i=0}^{K-2} \underbrace{\sum\nolimits_{j=2}^{\min\{q,K-i\}} \binom{K-i}{j}}_{={D}(K-i,q) \text{ in \eqref{eq:comp_merging}}},
\end{equation}
\noindent where the first summation in \eqref{eq:complexity} accounts for the maximum of $K-1$ merging of \emph{two} coalitions before the grand coalition forms and the second summation corresponds to the worst case number of iterations within the \texttt{while} loop in Step 3 of Algorithm \ref{alg:coalition} before exactly two coalitions merge. The binomial coefficient $\binom{K-i}{j}$ is the worst case number of iterations within the \texttt{for} loop before the \text{if} condition in Step 11 is true according to which coalition merging occurs. From Lemma \ref{lem:complexity} and observing that $i = 0$ in \eqref{eq:complexity} admits the largest growth for the upper bound, we obtain the complexity result.
\end{IEEEproof}

Accordingly, Algorithm \ref{alg:coalition} has polynomial time complexity for fixed $q$. The worst case number of iterations $W(K,q)$ in \eqref{eq:complexity} is illustrated in \figurename~\ref{fig:complexity}.%
\section{Numerical Results}\label{sec:simResults}%
\begin{figure}[t]
  \centering
  \includegraphics[width=8.8cm,clip]{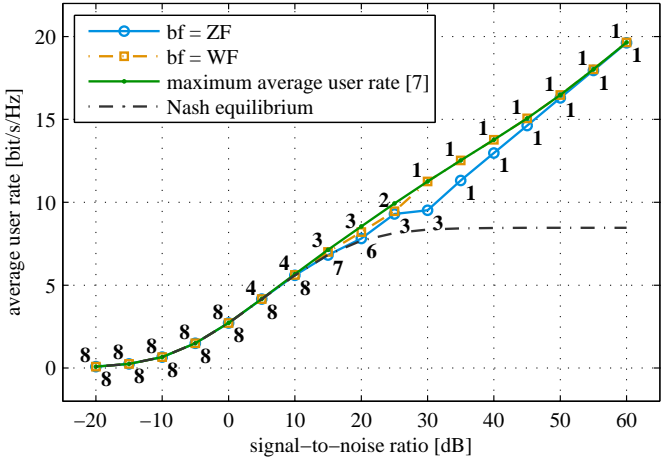}
  \caption{Average rate of the 8 users in the setting in \figurename~\ref{fig:topo} using Algorithm~\ref{alg:coalition}. The number under (above) the curve is the number of coalitions with ZF (WF).}\label{fig:rate_all}
\end{figure}
\begin{figure}[t]
  \centering
  \includegraphics[width=8.8cm,clip]{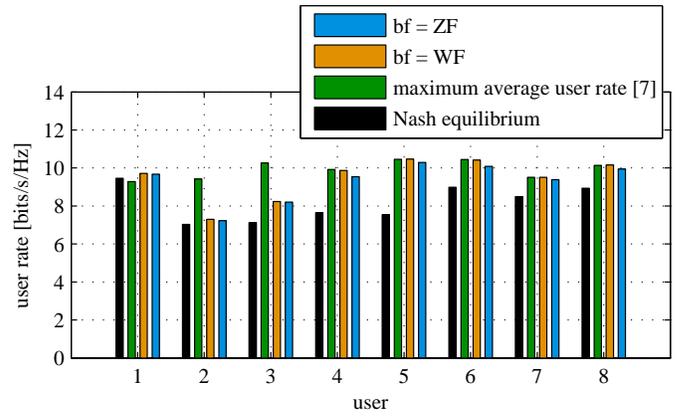}
  \caption{User rates at 25 dB SNR.}\label{fig:rate_user}
\end{figure}

We first consider the setting with $8$ links in \figurename~\ref{fig:topo}. Using the coalition formation algorithm in Algorithm~\ref{alg:coalition}, the average user rate is plotted for increasing SNR in \figurename~\ref{fig:rate_all} where we set $q = 8$ and $\mat{\epsilon} = \mat{0}$ as input to the algorithm. In the low SNR regime, single-player coalitions exist supporting Proposition \ref{thm:singleCoal}. Note that in the low SNR regime, the outcome with joint MRT is efficient \cite{Larsson2008a}. In the mid SNR regime, coalition formation improves the joint performance of the links from the Nash equilibrium. Optimal average user rate is obtained using the monotonic optimization algorithm from the supplementary material of \cite{Bjornson2013}. There, the beamforming space is not restricted to ZF or WF scheme. The average user rate in Nash equilibrium saturates in the high SNR regime. This is contrary to ZF and WF coalition formation where the average user rate increases linearly due to the formation of the grand coalition. The formation of the grand coalition in the high SNR regime supports Proposition \ref{thm:grandCoal}.

Comparing the WF and ZF schemes, it is observed in \figurename~\ref{fig:rate_all} that with WF precoding, larger coalitions form at lower SNR values than with ZF beamforming. As a result, higher average user rate gains are achieved with WF coalition formation than with ZF coalition formation. For example, at SNR = $25$ dB, the coalition structure with WF precoding is $\{\{1,3,4,5,6,7,8\},\{2\}\}$ while the coalition structure with ZF beamforming is $\{\{1,3\},\{2\},\{4,5,6,7,8\}\}$. Observe that the coalition $\{4,5,6,7,8\}$ is the cluster of links in the bottom right side of \figurename~\ref{fig:topo}. This set of links form a single coalition to reduce the interference between one another. The achieved user rates at $25$ dB SNR are shown in \figurename~\ref{fig:rate_user}. For user one, it is observed that optimal beamforming reduces his utility compared to the Nash equilibrium. Hence, the maximum sum rate beamforming strategy is not acceptable for player one and consequently the maximum sum rate strategy is not stable for joint cooperation. For both WF and ZF schemes, player two is in a single-player coalition. However, his rate in both schemes is higher than in Nash equilibrium. This reveals that in this case, the formed coalitions outside $\{2\}$ have positive effects on player two.

\begin{figure}[t]
  \centering
  \includegraphics[width=8.8cm,clip]{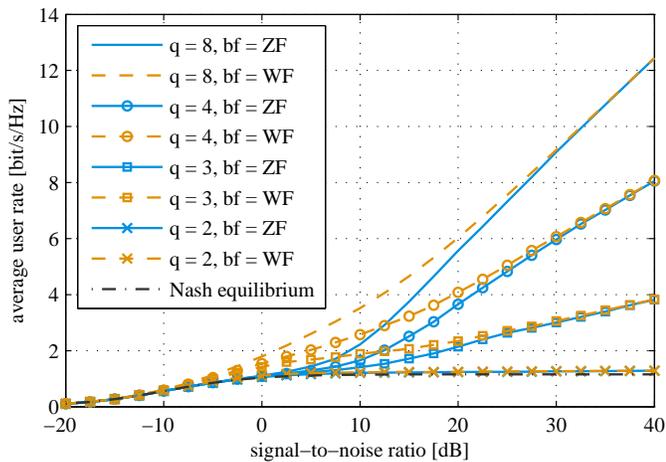}
  \caption{Comparison of average user rates achieved with Algorithm~\ref{alg:coalition} for different values of $q$. The overhead is set to $\mat{\epsilon} = \mat{0}$.}\label{fig:complexity_q_rate}
\end{figure}
\begin{figure}[t]
  \centering
  \includegraphics[width=8.8cm,clip]{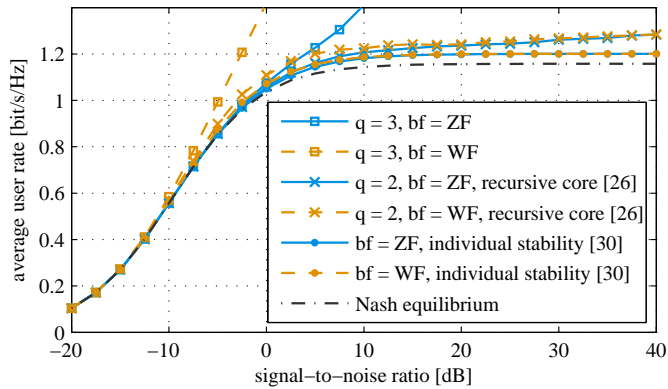}
  \caption{Comparison of average user rates achieved with Algorithm~\ref{alg:coalition} and coalition formation algorithms from \cite{Guruacharya2013} and \cite{Zhou2013}. The overhead is $\mat{\epsilon} = \mat{0}$.}\label{fig:complexity_q_rate2}
\end{figure}
In the following figures, we consider $8$ users with $8$ antennas each and we average the performance over $10^3$ random channel realizations. First, in \figurename~\ref{fig:complexity_q_rate} - \figurename~\ref{fig:complexity_q_iter} we compare the performance of WF and ZF coalition formation for different values of $q$. We also relate and compare our results to the coalition formation algorithms in \cite{Guruacharya2013} and \cite{Zhou2013}. Then, in \figurename~\ref{fig:complexity_o_rate} - \figurename~\ref{fig:complexity_o_iter} we plot the performance of coalition formation for two overhead models defined in \eqref{eq:alpha} and \eqref{eq:alpha2}.

In \figurename~\ref{fig:complexity_q_rate}, the average user rate is plotted for different values of $q$. The parameters $q$ influences the deviation rule as defined in Definition \ref{def:coalition_deviation} and specifies the largest number of coalitions which are allowed to merge in one iteration of Algorithm~\ref{alg:coalition}. It is shown that as $q$ increases, higher performance is obtained. Interestingly, although the number of iterations of Algorithm~\ref{alg:coalition} is not restricted, performance loss still occurs for smaller $q$.

\begin{figure}[t]
  \centering
  \includegraphics[width=8.8cm,clip]{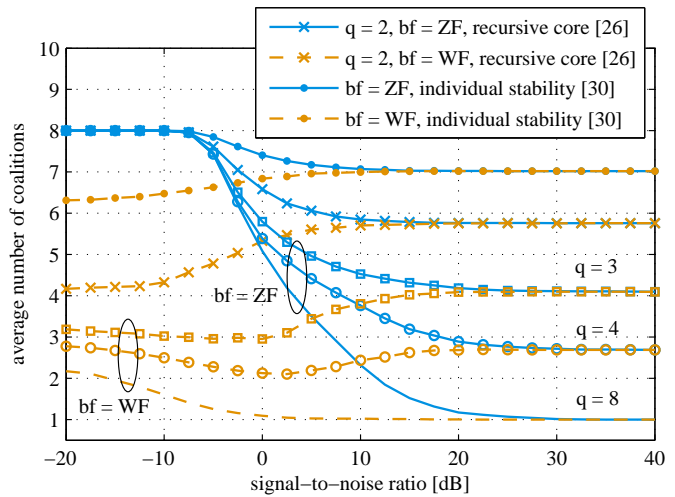}
  \caption{Comparison of the average number of coalitions obtained by Algorithm~\ref{alg:coalition} for different values of $q$.}\label{fig:complexity_q_size}
\end{figure}
\begin{figure}[t]
  \centering
  \includegraphics[width=8.8cm,clip]{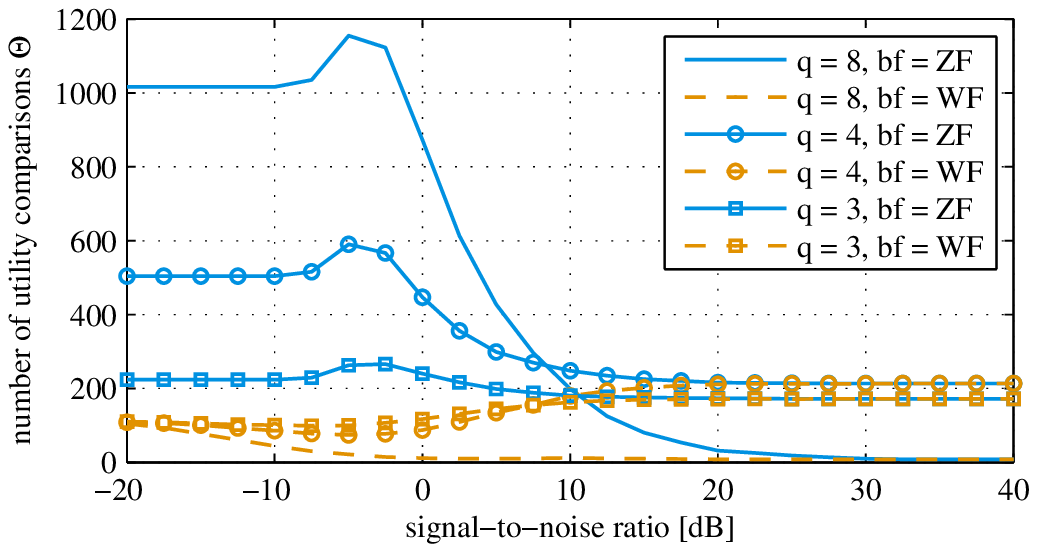}
  \caption{Comparison of the average number of utility comparisons $\Theta$ required by Algorithm~\ref{alg:coalition} for different values of $q$.}\label{fig:complexity_q_iter}
\end{figure}
\begin{figure}[t]
  \centering
  \includegraphics[width=8.8cm,clip]{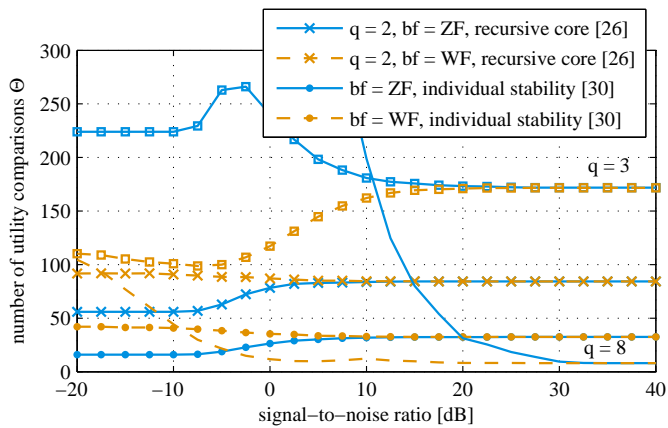}
  \caption{Comparison of the average number of utility comparisons $\Theta$ required by Algorithm~\ref{alg:coalition} for different values of $q$.}\label{fig:complexity_q_iter2}
\end{figure}

In \figurename~\ref{fig:complexity_q_rate2}, we compare our algorithm with the algorithms in \cite{Guruacharya2013} and \cite[Algorithm 1]{Zhou2013}. Note that the system models and cooperation models in \cite{Guruacharya2013} and \cite{Zhou2013} are different than ours. In \cite{Guruacharya2013} cooperation between a set of base stations is based on network MIMO schemes which require exchange of user data between the transmitter. In \cite{Zhou2013}, a MIMO interference channel is considered and cooperation in a coalition is according to ZF transmission. In the proposed coalition formation algorithm in \cite{Guruacharya2013}, the deviation rule is based on merging of two coalitions, i.e., corresponds to $2$-Deviation in Definition \ref{def:coalition_deviation}, and the comparison relation is according to $2$-Dominance in Definition \ref{def:dominance}. Hence, the algorithm in \cite{Guruacharya2013} corresponds to our algorithm with $q=2$, and the authors prove that the resulting coalition structure lies in the recursive core of their coalition formation game. In \figurename~\ref{fig:complexity_q_rate2}, the average performance of the links for $q = 2$ is less than for $q = 3$ and is higher than for \cite[Algorithm 1]{Zhou2013} which is based on individual stability. In \figurename~\ref{fig:complexity_q_rate2}, it is shown that the average performance of individual stability is slightly higher than in Nash equilibrium. Note however that the performance of \cite[Algorithm 1]{Zhou2013} can be significantly improved for sufficiently large number of antennas at the transmitters.

\figurename~\ref{fig:complexity_q_rate} and \figurename~\ref{fig:complexity_q_rate2} lead us to the conclusion that in multi-antenna interference channels, performance improvement through cooperation with ZF or WF beamforming depends greatly on the number of users that are allowed to deviate and cooperate at a time.

In \figurename~\ref{fig:complexity_q_size}, the average number of coalitions obtained from Algorithm~\ref{alg:coalition} are plotted and compared to \cite[Algorithm 1]{Zhou2013}. The average number of coalitions with the ZF beamforming scheme is larger than with the WF beamforming scheme, i.e., with WF more coalitions merge than with ZF. At low SNR, single player coalitions exist with ZF beamforming. This illustrates the result in Proposition~\ref{thm:singleCoal}. Since, WF beamforming converges to MRT beamforming at low SNR, we observe that in this SNR regime some coalitions form with WF beamforming. As $q$ increases, the average number of coalitions decreases. This explains the performance loss in \figurename~\ref{fig:complexity_q_rate}. Both, WF and ZF coalition formation obtain the same number of coalitions for the same $q$ at high SNR. This is because WF beamforming converges to ZF beamforming in this SNR regime. Coalition formation with \cite[Algorithm 1]{Zhou2013} leads to high average number of coalitions which means that the number of users that cooperate is small.

In \figurename~\ref{fig:complexity_q_iter}, the average number of utility comparisons during Algorithm~\ref{alg:coalition} is plotted for different $q$. Note, that Algorithm~\ref{alg:coalition} starts in single-player coalitions and the users initially search for the largest number of coalitions to merge. The number of utility comparisons $\Theta$ during the search is specified in Step 8 in Algorithm~\ref{alg:coalition} and corresponds to the total number of utility comparisons the users need to do before the algorithm converges. It is shown in \figurename~\ref{fig:complexity_q_iter} that the WF scheme requires generally much lower number of utility comparisons than the ZF scheme. The reason for this is with the WF scheme more coalitions merge than with the ZF scheme according to \figurename~\ref{fig:complexity_q_size} which leads to faster convergence rate of Algorithm~\ref{alg:coalition} with WF beamforming than with the ZF scheme. The number of utility comparisons $\Theta$ generally decreases for smaller $q$ since the number of possible deviations decreases. For a specific $q$, it is observed that $\Theta$ is not monotonic with SNR. With ZF beamforming, the number of utility comparisons increases at around $-5$ dB SNR and afterwards decreases. The reason for this behavior is at very low SNR no coalitions merge and hence Algorithm~\ref{alg:coalition} terminates after searching over all merging possibilities of single-player coalitions. At around $-5$ dB SNR, a small number of coalitions merge and hence Algorithm~\ref{alg:coalition} requires additional iterations to search for possible merging in the new coalition structure. Since the number of coalitions that merge is small, the number of comparisons is large. At larger SNR, the number of coalitions that merge increases and hence the coalitions are found faster than at smaller SNR. At high SNR, the grand coalition is then favorable for all the users and for $q=8$ it is possible to build the grand coalition in a single iteration of Algorithm~\ref{alg:coalition}. This is contrary to the case $q<8$ which explains the performance advantage of $q = 8$. In \figurename~\ref{fig:complexity_q_iter2}, we can see that \cite[Algorithm 1]{Zhou2013} requires a low number of utility comparisons and hence has low complexity.%

\begin{figure}[t]
  \centering
  \includegraphics[width=8.8cm,clip]{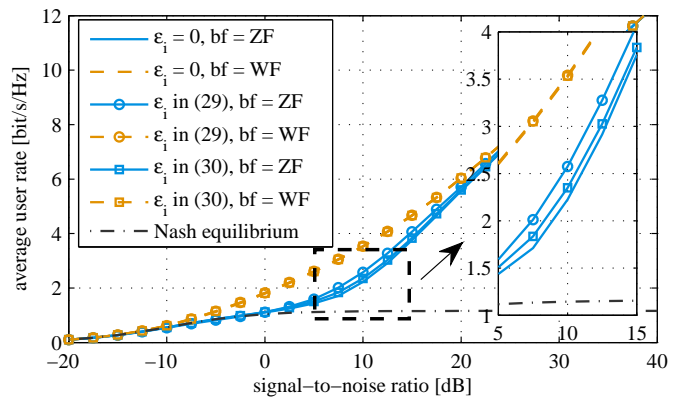}
  \caption{Comparison of average user rates of Algorithm~\ref{alg:coalition} with $q = 8$ for the overhead models specified in \eqref{eq:alpha} and \eqref{eq:alpha2}.}\label{fig:complexity_o_rate}
\end{figure}
\begin{figure}[t]
  \centering
  \includegraphics[width=8.8cm,clip]{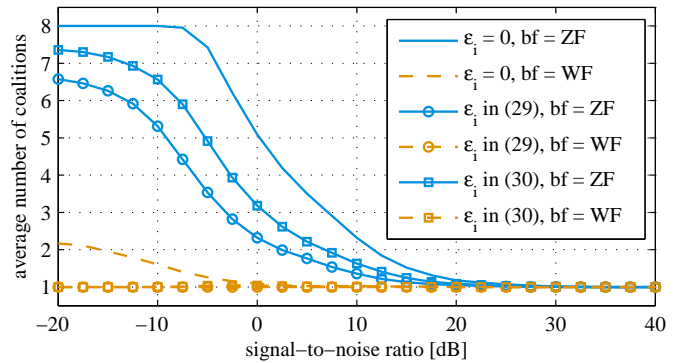}
  \caption{Comparison of average number of coalitions resulting from Algorithm~\ref{alg:coalition} with $q = 8$ for the overhead models specified in \eqref{eq:alpha} and \eqref{eq:alpha2}.}\label{fig:complexity_o_size}
\end{figure}
\begin{figure}[t]
  \centering
  \includegraphics[width=8.8cm,clip]{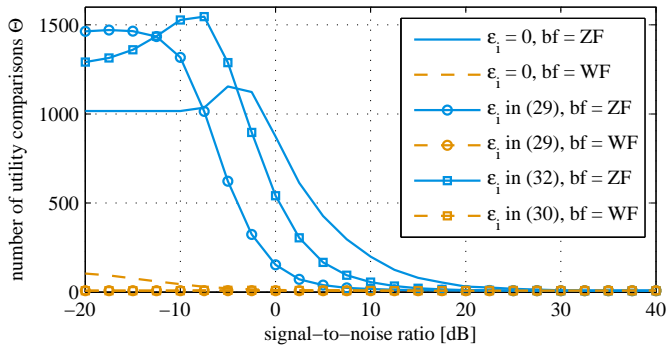}
  \caption{Comparison of average number of utility comparisons $\Theta$ in Algorithm~\ref{alg:coalition} with $q = 8$ for the overhead models specified in \eqref{eq:alpha} and \eqref{eq:alpha2}.}\label{fig:complexity_o_iter}
\end{figure}
%

For the subsequent plots, we define two overhead models to specify $\mat{\epsilon}$ for a given coalition structure $\coCS$ as follows:
\begin{align}\label{eq:alpha}
\epsilon_i & = \frac{\abs{\coS}}{\abs{\coN}} u_i(\bw_1^\text{bf},\ldots,\bw_K^\text{bf}), \quad i \in \coS, \coS \in \coCS,\\ \label{eq:alpha2}
\epsilon_i & = \frac{1}{\abs{\coN}} u_i(\bw_1^\text{bf},\ldots,\bw_K^\text{bf}), \quad \text{for all } i \in \coN,
\end{align}
\noindent with $\text{bf} = \{\text{\small{ZF},\small{WF}}\}$. According to \eqref{eq:alpha} and following the comparison relation in \eqref{eq:preference}, the overhead for a user which examines whether to join a coalition $\coS$ increases with the size of $\coS$. In \eqref{eq:alpha2}, the overhead is assumed to be fixed. In both overhead models in \eqref{eq:alpha} and \eqref{eq:alpha2}, we assume that the overhead is proportional to the utility in the grand coalition since it corresponds to the largest overhead needed to examine whether the grand coalition forms.

In \figurename~\ref{fig:complexity_o_rate} - \figurename~\ref{fig:complexity_o_iter}, we set $q = 8$ and plot the performance of Algorithm~\ref{alg:coalition}. In \figurename~\ref{fig:complexity_o_rate}, the average user rate in both ZF and WF schemes is similar with both overhead models in \eqref{eq:alpha} and \eqref{eq:alpha2} and no overhead ($\epsilon_i = 0$). With ZF beamforming, it is shown that higher average user rate improvement is achieved for larger overhead. The explanation of this effect is that with larger overhead more coalitions merge as is shown in \figurename~\ref{fig:complexity_o_size}. With WF coalition formation, it is shown that with the proposed overhead models, the grand coalition always forms in the simulation setup.

In \figurename~\ref{fig:complexity_o_iter}, the average number of utility comparisons $\Theta$ required during Algorithm~\ref{alg:coalition} is plotted. With WF beamforming, $\Theta$ is very low which reveals that the WF beamforming scheme requires less complexity in coalition formation than ZF beamforming. Generally, with ZF beamforming and higher overhead, the number of utility comparisons decrease. This is observed after approximately $-2.5$ dB SNR. At low SNR, the number of utility comparisons are larger including the overhead compared to no overhead because a small number of users merge to form coalitions as is shown in \figurename~\ref{fig:complexity_o_size}. Consequently, a larger number of searches during Algorithm~\ref{alg:coalition} is required. In comparison to no overhead, it can be observed that the maximum number of utility comparisons shifts to the left on the SNR axis as the overhead increases.%
\section{Conclusions}\label{sec:conc}
We study cooperation in the MISO IFC using coalitional games. A transmitter's noncooperative transmission strategy is MRT. We consider ZF or WF precoding for cooperative user transmission in a coalition. The necessary and sufficient conditions under which all users profit from joint cooperation with ZF are characterized. It is shown that incorporating an overhead in deviation of a coalition enlarges the set of conditions under which the users have the incentive to cooperate compared to the case without considering the overhead. Afterwards, we consider coalitional formation games and propose a distributed coalition formation algorithm based on coalition merging. The complexity of the algorithm is tuneable and its implementation requires two-bit signalling between the transmitters. We prove that the algorithm terminates in polynomial time and generates a coalition structure in the coalition structure stable set. The performance improvement of the links through cooperation with simple distributed transmission schemes is achieved and depends on the complexity for coalition deviation and user-specific overhead measures.

As future work, we will extend the current work to consider multiple antennas at both the transmitters and receivers. Moreover, we will generalize to multi cell settings in which multiple receivers are associated with each transmitter. Significant to study for these settings is the appropriate design for the cooperative precoding schemes used within the coalition. Nevertheless, our interest is also to investigate different solution concepts for coalitional games in partition form and also to draw comparisons between them.

\appendices
\section{Proof of Proposition \ref{thm:grandCoal}}\label{proof:grandCoal}
Considering an arbitrary player $i$ in an arbitrary coalition $\coS$, we write the condition in \eqref{eq:GrandCoalZF1} as
\begin{multline}\label{eq:GrandCoalZF}
\log_2\pp{1 + \frac{\sabs{\bh_{ii}^\H \bwzf{i}{\coS}}}{\sigma^2 + \sum_{j \in \coN \backslash \coS} \sabs{\bh_{ji}^\H \bw_j^\mrt}}} - \epsilon_i \\ \leq \log_2\pp{1 + \frac{ \sabs{\bh_{ii}^\H \bwzf{i}{\coN}}}{\sigma^2}},
\end{multline}
\noindent which is equivalent to
\begin{equation}\label{eq:quad0}
1 + \frac{ A_{i,\coS} }{\sigma^2 + B_{i,\coS} } \leq 2^{\epsilon_i} + 2^{\epsilon_i} \frac{ C_{i} }{\sigma^2},
\end{equation}
\noindent where the introduced notations $A_{i,\coS},B_{i,\coS},$ and $C_{i}$ are given in \eqref{eq:A} and \eqref{eq:B} in Proposition \ref{thm:grandCoal}. Cross multiplying \eqref{eq:quad0} and solving for $\sigma^2$ we get the following condition
\begin{equation}\label{eq:quad}
f(\sigma^2) \geq 0, \quad \sigma^2 > 0,
\end{equation}
\noindent where $f(\sigma^2) := (2^{\epsilon_i} - 1)(\sigma^2)^2 + (2^{\epsilon_i} (B_{i,\coS} + C_{i}) - (B_{i,\coS} + A_{i,\coS})) \sigma^2 + 2^{\epsilon_i} C_{i} B_{i,\coS}$. In order to analyze \eqref{eq:quad}, a case study is summarized in Table \ref{tbl:case_study} and illustrated in \figurename~\ref{fig:illust}. 

If $\epsilon_i = 0$, then $f(\sigma^2)$ in \eqref{eq:quad} is a straight line as in Case I in \figurename~\ref{fig:illust}. The condition in \eqref{eq:quad} reduces to
\begin{align}
(C_{i} - A_{i,\coS}) \sigma^2 + C_{i} B_{i,\coS} & \geq 0,\\
\sigma_0^2 :=\frac{C_{i} B_{i,\coS}}{A_{i,\coS} - C_{i} } & \geq \sigma^2 > 0,
\end{align}
and corresponds to the shaded region in \figurename~\ref{fig:illust}.

If $\epsilon_i > 0$, and having $(2^{\epsilon_i} - 1) > 0$ then the quadratic polynomial $f(\sigma^2)$ in \eqref{eq:quad} describes a parabola with a minimum and opens upwards as illustrated in Cases II-IV in \figurename~\ref{fig:illust}. If the minimum of the parabola is strictly larger than zero (Case II in \figurename~\ref{fig:illust}), then the quadratic equation has no real roots and the discriminant of $f(\sigma^2)$ is negative, i.e.,
\begin{multline}\label{eq:delta}
\Delta_{i,\coS} := (2^{\epsilon_i} (B_{i,\coS} + C_{i}) - (B_{i,\coS} + A_{i,\coS}))^2 \\ - 4 (2^{\epsilon_i} - 1)2^{\epsilon_i} C_{i} B_{i,\coS} < 0.
\end{multline}
If \eqref{eq:delta} is satisfied, so is the condition in \eqref{eq:quad} for any $0 < \sigma^2 < \infty$ because the entire parabola has strictly positive values.

\begin{figure}
  \centering
\scalebox{0.9} 
{
\begin{pspicture}(0.7,-4)(9,4.62)
\definecolor{color1643b}{rgb}{0.8,0.8,0.8}
\psframe[linewidth=0.04,linecolor=color1643b,dimen=outer,fillstyle=solid,fillcolor=color1643b](9.083125,-0.7)(8.163125,-3.4)
\psframe[linewidth=0.04,linecolor=color1643b,dimen=outer,fillstyle=solid,fillcolor=color1643b](3.183125,3.8)(1.283125,1.1)
\psline[linewidth=0.04cm,arrowsize=0.04cm 2.0,arrowlength=2.0,arrowinset=0.0]{<-}(1.283125,4.2)(1.283125,1.0)
\psline[linewidth=0.04cm,arrowsize=0.04cm 2.0,arrowlength=2.0,arrowinset=0.0]{->}(0.683125,1.2)(4.683125,1.2)
\usefont{T1}{ptm}{m}{n}
\rput(1.1067188,0.995){\footnotesize $0$}
\usefont{T1}{ptm}{m}{n}
\rput(4.266719,0.895){\footnotesize $\sigma^2$}
\usefont{T1}{ptm}{m}{n}
\rput(0.80671877,3.795){\footnotesize $f(\sigma^2)$}
\psframe[linewidth=0.04,linecolor=color1643b,dimen=outer,fillstyle=solid,fillcolor=color1643b](9.083125,3.8)(5.983125,1.1)
\psline[linewidth=0.04cm,arrowsize=0.04cm 2.0,arrowlength=2.0,arrowinset=0.0]{<-}(5.983125,4.2)(5.983125,1.0)
\psline[linewidth=0.04cm,arrowsize=0.04cm 2.0,arrowlength=2.0,arrowinset=0.0]{->}(5.383125,1.2)(9.383125,1.2)
\usefont{T1}{ptm}{m}{n}
\rput(5.806719,0.995){\footnotesize $0$}
\psbezier[linewidth=0.04](6.583125,3.9)(6.583125,0.8)(8.483125,0.8)(8.483125,3.9)
\usefont{T1}{ptm}{m}{n}
\rput(8.966719,0.895){\footnotesize $\sigma^2$}
\usefont{T1}{ptm}{m}{n}
\rput(5.5067186,3.795){\footnotesize $f(\sigma^2)$}
\psline[linewidth=0.04cm](1.683125,4.1)(3.483125,0.6)
\usefont{T1}{ptm}{m}{n}
\rput(3.0267189,0.895){\footnotesize $\sigma_0^2$}
\psframe[linewidth=0.04,linecolor=color1643b,dimen=outer,fillstyle=solid,fillcolor=color1643b](4.383125,-0.7)(3.083125,-3.4)
\psline[linewidth=0.04cm,arrowsize=0.04cm 2.0,arrowlength=2.0,arrowinset=0.0]{<-}(3.083125,-0.3)(3.083125,-3.5)
\psline[linewidth=0.04cm,arrowsize=0.04cm 2.0,arrowlength=2.0,arrowinset=0.0]{->}(0.683125,-3.3)(4.683125,-3.3)
\usefont{T1}{ptm}{m}{n}
\rput(2.9067187,-3.505){\footnotesize $0$}
\psbezier[linewidth=0.04](0.883125,-1.4)(0.883125,-4.5)(2.783125,-4.5)(2.783125,-1.4)
\usefont{T1}{ptm}{m}{n}
\rput(4.266719,-3.605){\footnotesize $\sigma^2$}
\usefont{T1}{ptm}{m}{n}
\rput(2.6067188,-0.705){\footnotesize $f(\sigma^2)$}
\psframe[linewidth=0.04,linecolor=color1643b,dimen=outer,fillstyle=solid,fillcolor=color1643b](6.923125,-0.7)(5.983125,-3.4)
\psline[linewidth=0.04cm,arrowsize=0.04cm 2.0,arrowlength=2.0,arrowinset=0.0]{<-}(5.983125,-0.3)(5.983125,-3.5)
\psline[linewidth=0.04cm,arrowsize=0.04cm 2.0,arrowlength=2.0,arrowinset=0.0]{->}(5.383125,-3.3)(9.383125,-3.3)
\usefont{T1}{ptm}{m}{n}
\rput(5.806719,-3.505){\footnotesize $0$}
\psbezier[linewidth=0.04](6.583125,-1.5)(6.583125,-4.6)(8.483125,-4.6)(8.483125,-1.5)
\usefont{T1}{ptm}{m}{n}
\rput(8.966719,-3.605){\footnotesize $\sigma^2$}
\usefont{T1}{ptm}{m}{n}
\rput(5.5067186,-0.705){\footnotesize $f(\sigma^2)$}
\usefont{T1}{ptm}{m}{n}
\rput(6.826719,-3.605){\footnotesize $\sigma_1^2$}
\usefont{T1}{ptm}{m}{n}
\rput(8.226719,-3.605){\footnotesize $\sigma_2^2$}
\usefont{T1}{ptm}{m}{n}
\rput(2.6321876,4.41){Case I}
\usefont{T1}{ptm}{m}{n}
\rput(7.464375,-0.09){Case IV}
\usefont{T1}{ptm}{m}{n}
\rput(7.4921875,4.41){Case II}
\usefont{T1}{ptm}{m}{n}
\rput(2.5521874,-0.09){Case III}
\psline[linewidth=0.04cm](0.483125,0.4)(9.383125,0.4)
\psline[linewidth=0.04cm](4.883125,4.6)(4.883125,-3.9)
\end{pspicture}
}
  \caption{Illustrations for the case study of the quadratic inequality in \eqref{eq:quad}.}\label{fig:illust}
\end{figure}
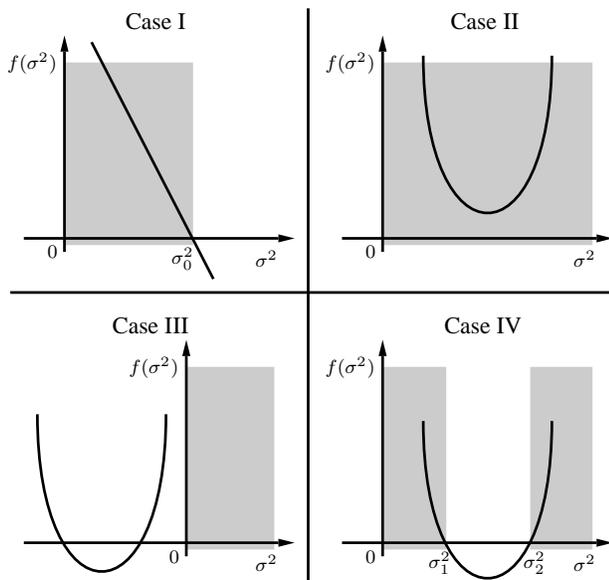
\begin{table}[t]
\centering
\caption{\label{tbl:case_study}Case study for analysing the quadratic equation in \eqref{eq:quad}.}
\begin{tabular}{|l|ccccc|}
   \hline
  Case I     & $\epsilon_i = 0$ & & & & \\ \hline
  Case II    & $\epsilon_i > 0$ & and &$\Delta_{i,\coS} < 0$ & & \\ \hline
  Case III    & $\epsilon_i > 0$ & and &$\Delta_{i,\coS} \geq 0$ & and &$\sigma_1^2 + \sigma_2^2 \leq 0$ \\ \hline
  Case IV   & $\epsilon_i > 0$ & and &$\Delta_{i,\coS} \geq 0$ & and &$\sigma_1^2 + \sigma_2^2 > 0$ \\
  \hline
\end{tabular}
\end{table}

For $\Delta_{i,\coS}\geq 0$, $f(\sigma^2)$ in \eqref{eq:quad} has two real roots which corresponds to Cases III and IV in \figurename~\ref{fig:illust}. Then condition \eqref{eq:quad} holds for
\begin{equation}\label{eq:condition_grancoal_player}
\sigma^2 \leq \sigma^2_1 \text{  or  } \sigma^2\geq \sigma^2_2 \text{  and  } \sigma^2>0,
\end{equation}
\noindent where $\sigma^2_1$ and $\sigma^2_2$ are the roots of $f(\sigma^2)$ in \eqref{eq:quad} given as
\begin{equation}\label{eq:roots}
    \sigma_1^2 = \frac{-\Psi_{i,\coS} - \sqrt{\Delta_{i,\coS}} }{2(2^{\epsilon_i} - 1)}, \quad \sigma_2^2 =\frac{-\Psi_{i,\coS} + \sqrt{\Delta_{i,\coS}} }{2(2^{\epsilon_i} - 1)}.
\end{equation}
\noindent with $\Psi_{i,\coS}:= \pp{2^{\epsilon_i} (B_{i,\coS} + C_{i}) - (B_{i,\coS} + A_{i,\coS})} \geq 0$. The product of both roots is \cite[3.8.1]{Abramowitz1972}: $\sigma^2_1 \sigma^2_2 = \pp{2^{\epsilon_i} C_{i} B_{i,\coS}}/\pp{2^{\epsilon_i} - 1} \geq 0$. Thus, both roots have the same sign. In order to determine whether both roots are negative or positive we study their sum. The sum of the roots is less than or equal to zero if and only if
\begin{equation}
\Psi_{i,\coS}:= \pp{2^{\epsilon_i} (B_{i,\coS} + C_{i}) - (B_{i,\coS} + A_{i,\coS})} \geq 0.
\end{equation}
\noindent Under the above condition, \eqref{eq:condition_grancoal_player} is satisfied for any $0 < \sigma^2 < \infty$ because the largest root is negative.

In Case IV, the sum of the roots in \eqref{eq:roots} is strictly positive. This is satisfied if and only if
\begin{align}
\log_2\pp{B_{i,\coS} + A_{i,\coS}} - \log_2\pp{B_{i,\coS} + C_{i}} > {\epsilon_i}.
\end{align}
\noindent In this case, the condition in \eqref{eq:condition_grancoal_player} is illustrated in the shaded areas in Case IV in \figurename~\ref{fig:illust}.

Note that Cases I-IV are all possibilities to study \eqref{eq:quad} associated with a user $i$ in coalition $\coS$. Combining Cases I-IV, we formulate the lower and upper bounds on $\sigma^2$ for a specific user $i$ in a coalition $\coS$ in \eqref{eq:sigma_lb} and \eqref{eq:sigma_ub}, respectively. Since the derived conditions in \eqref{eq:condition_grancoal_player} must hold for all players and all coalitions, we must take the maximum over all players and all coalitions on the lower bound $\sigma^2_2$ (largest root of $f(\sigma^2)$ in \eqref{eq:quad}) and the minimum over all players and coalitions for the upper bound $\sigma^2_1$ (smallest root of $f(\sigma^2)$).%
\bibliographystyle{IEEEtran}
\bibliography{references}
\end{document}